\documentclass[a4paper
               ,11pt
               ,oneside
               ,aps
               ,notitlepage
               ,nofootinbib
               ,superscriptaddress
               ,tightenlines
               ]
               {revtex4}

\usepackage{amsmath}
\usepackage{amssymb}  
\usepackage{amsthm}
\usepackage{amsfonts}
\usepackage{graphicx}
\usepackage{bbm}

\usepackage{hyperref}
\hypersetup{
    colorlinks,
    citecolor=black,
    filecolor=black,
    linkcolor=black,
    urlcolor=black
}

\addtocontents{toc}{\vspace{-0.3cm}}

\theoremstyle{plain}               
\newtheorem{thm}{Theorem}
\newtheorem{lem}{Lemma}
\newtheorem{cor}{Corollary}

\newtheorem{defn}{Definition}
\newtheorem{exmp}{Example}[section]
\theoremstyle{remark}

\newcommand{\norm}[2]{\ensuremath{|\!|#1|\!|_{#2}}}

\newcommand{\tr}{\textnormal{tr}}

\newcommand{\Trace}[1]{\ensuremath{\tr \left( #1 \right)}}

\newcommand{\ket}[1]{| #1 \rangle}

\newcommand{\bra}[1]{\langle #1 |}

\newcommand{\braket}[2]{\langle #1 | #2 \rangle}

\newcommand{\proj}[2]{| #1 \rangle\!\langle #2 |}

\newcommand{\id}{\ensuremath{\mathbbm{1}}}

\def\beq{\begin{equation}}
\def\eeq{\end{equation}}
\def\bq{\begin{quote}}
\def\eq{\end{quote}}
\def\ben{\begin{enumerate}}
\def\een{\end{enumerate}}
\def\bit{\begin{itemize}}
\def\eit{\end{itemize}}

\def\ra{\rightarrow}

\def\lb{\left(}
\def\rb{\right)}
\def\lset{\lbrace}
\def\rset{\rbrace}

\def\l|{\left|}
\def\r|{\right|}
\def\lbr{\left[}
\def\rbr{\right]}
\def\ident{\textnormal{id}}
\def\one{\id}

\newcommand\C{\mathbbm{C}}

\newcommand\R{\mathbbm{R}}
\newcommand\N{\mathbbm{N}}
\newcommand\M{\mathcal{M}}
\newcommand\D{\mathcal{D}}

\newcommand{\U}{\mathcal{U}}

\newcommand{\Tm}{\mathcal{T}}
\newcommand{\Rm}{\mathcal{R}}
\newcommand{\Sm}{\mathcal{S}}
\newcommand{\Em}{\mathcal{E}}
\newcommand{\Dm}{\mathcal{D}}
\newcommand{\Cm}{\mathcal{C}}
\newcommand{\Lm}{\mathcal{L}}

\newcommand{\Pm}{\mathcal{P}}

\newcommand{\Um}{\mathcal{U}}
\newcommand{\Vm}{\mathcal{V}}
\newcommand{\Wm}{\mathcal{W}}

\begin{document}

\title{\vspace{-1.0cm}{\textbf{Positivity of Linear Maps under Tensor Powers}}}
\author{Alexander M\"uller-Hermes}
\email{muellerh@ma.tum.de}
\affiliation{Zentrum Mathematik, Technische Universit\"{a}t M\"{u}nchen, 85748 Garching, Germany}
\author{David Reeb}
\email{reeb.qit@gmail.com}
\affiliation{Institute for Theoretical Physics, Leibniz Universit\"at Hannover, 30167 Hannover, Germany}
\affiliation{Zentrum Mathematik, Technische Universit\"{a}t M\"{u}nchen, 85748 Garching, Germany}

\author{Michael M. Wolf}
\email{m.wolf@tum.de}
\affiliation{Zentrum Mathematik, Technische Universit\"{a}t M\"{u}nchen, 85748 Garching, Germany}

\date{\today}

\begin{abstract}
We investigate linear maps between matrix algebras that remain positive under tensor powers, i.e., under tensoring with $n$ copies of themselves. Completely positive and completely co-positive maps are trivial examples of this kind. We show that for every $n\in\mathbb{N}$ there exist non-trivial maps with this property and that for two-dimensional Hilbert spaces there is no non-trivial map for which this holds for all $n$.  For higher dimensions we reduce the existence question of such non-trivial ``tensor-stable positive maps''  to a one-parameter family of maps and show that an affirmative answer would imply the existence of NPPT bound entanglement.
 
 As an application we show that any tensor-stable positive map that is  not completely positive yields an upper bound on the quantum channel capacity, which for the transposition map gives the well-known cb-norm bound. We furthermore show that the latter is an upper bound even for the LOCC-assisted quantum capacity, and that moreover it is a strong converse rate for this task.
\end{abstract}

\vspace*{-0.3cm}
\maketitle

\vspace*{-1cm}

\tableofcontents

\section{Introduction and main results}

Within the set $\M_d$ of complex $d\times d$-matrices we denote the cone of positive matrices by $\M^+_d$ (we call ``positive semidefinite matrices'' simply ``positive matrices''). A linear map $\Pm:\M_{d_1}\ra \M_{d_2}$ is called positive if $\Pm\lb\M^+_{d_1}\rb\subseteq\M^+_{d_2}$, and we then write $\Pm\geq0$. We want to study how positivity of a linear map behaves when taking tensor powers. Therefore we consider the following: 

\begin{defn}[Tensor-stable positivity]\hfill
 
\begin{enumerate}
\item[(i)]
A linear map $\Pm:\M_{d_1}\ra \M_{d_2}$ is called \textbf{n-tensor-stable positive} for some number $n\in\N$ if the map $\Pm^{\otimes n}:\M_{d^n_1}\ra\M_{d^n_2}$ is positive. 
\item[(ii)]A linear map $\Pm:\M_{d_1}\ra \M_{d_2}$ is called \textbf{tensor-stable positive} if the map $\Pm$ is $n$-tensor-stable positive for all $n\in\N$.
\end{enumerate} 
\end{defn}

Note that every $n$-tensor-stable positive map is in particular a positive map. The following example displays some maps that are easily seen to be tensor-stable positive. We will call all maps from these classes \textbf{trivial tensor-stable positive maps}.

\begin{exmp}[Trivial tensor-stable positive maps]\hfill 
\begin{enumerate}
\item All completely positive maps are tensor-stable positive, i.e. all linear maps $\Tm:\M_{d_1}\ra \M_{d_2}$ such that $(\ident_{d}\otimes \Tm): \M_{d}\otimes \M_{d_1}\ra \M_{d}\otimes \M_{d_2}$ is positive for all dimensions $d\in\N$.
\item All maps of the form $\vartheta_{d_2}\circ \Tm$ for a completely positive map $\Tm:\M_{d_1}\ra\M_{d_2}$ and the transposition $\vartheta_d:\M_d\ra\M_d$ are tensor-stable positive. The maps of this form are called \textbf{completely co-positive}.
\end{enumerate}
\label{exmp:1}
\end{exmp}

We will be concerned with three basic questions:
\begin{enumerate}
\item Are there any non-trivial tensor-stable positive maps?
\item How far away can an $n$-tensor-stable positive map be from the cones of completely positive and completely co-positive maps (i.e. from the two cones of trivial tensor-stable positive maps from Example \ref{exmp:1})?
\item What are the implications of question 1. for quantum information theory?
\end{enumerate}

Our main results are the following. In section \ref{sec:FamiliesN} we use (non-orthogonal) unextendible product bases to show:
\begin{thm}[Existence of $n$-tensor-stable positive maps]\label{ntensor}
For any $n\in\N$ and any $d_1,d_2 \geq 2$ there exists an $n$-tensor-stable positive map $\Pm:\M_{d_1}\ra\M_{d_2}$ that is \emph{not} a trivial tensor-stable positive map.

\end{thm}
Our construction used to obtain this theorem does not seem to suffice for constructing a non-trivial tensor-stable positive map (i.e.\ one for all $n\in\N$), and at the time of writing we do not know whether such a map exists. 

In section \ref{sec:Appl} we discuss applications and implications of tensor-stable positive maps for quantum information theory. We show that the existence of an $\infty$-locally entanglement annihilating channel~\cite{moravvcikova2010entanglement,filippov2012local,filippov2013bipartite} which is not entanglement breaking~\cite{horodecki2003entanglement} implies the existence of non-trivial tensor-stable positive maps. A quantum channel is called $\infty$-locally entanglement annihilating if any state when sent through arbitrarily many copies of the channel becomes fully separable. It is currently not known whether such channels exist outside the set of entanglement breaking channels~\citep{filippov2013dissociation}.

In Section \ref{capacityboundsubsection} we generalize the well-known transposition bound~\cite{holevo2001evaluating} to show that tensor-stable positive, but not completely positive, maps yield upper bounds on the quantum channel capacity as well as strong converse rates for this task (Section \ref{strongconverseTSsubsection}). In Section \ref{twowaysubsection} we show that the transposition bound is an upper bound even on the LOCC-assisted quantum capacity (see also Corollary \ref{cor:Q20}) and constitutes a strong converse rate for this task.


In light of these implications, deciding question 1.\ would have important consequences for quantum information theory. Whereas we cannot resolve this question in general, in section \ref{sec:Dist} we use techniques from the theory of entanglement distillation and a generalization of a technique used in~\cite{stormer2010tensor} to prove: 

\begin{thm}[Only trivial tensor-stable positive maps in $d=2$]

There are \emph{no} non-trivial tensor-stable positive maps $\Pm:\M_2\ra\M_d$ or $\Pm:\M_d\ra\M_2$ for any $d\in\N$. 
\label{thm:NoQubitMap}
\end{thm}

Furthermore, a non-trivial tensor-stable positive map exists iff one exists within the following one-parameter families based on Werner states~\cite{werner1989quantum}:

\begin{thm}[One-parameter family of candidates for non-trivial tensor-stable positivity]

Let $d_1,d_2\in\N$, $d\in\lset d_1,d_2\rset$, and for $p\in[-1,1]$ let
\begin{align}
\Pm_p := \Wm_p \otimes \big(\vartheta_d\circ\Wm_p\big)\,:\,~\M_d\otimes\M_d\to\M_d\otimes\M_d\,,
\label{equ:OneParamFamMap}
\end{align}
where we define for $X\in\M_d$:
\begin{align}\label{wernerchannel}
\Wm_p\lb X\rb := \frac{1}{d^2 - 1}\lb \lb d - p\rb\Trace{X} \id_d - \lb 1 - dp \rb X^T\rb.
\end{align}

\begin{itemize}

\item[(i)]If there exists a non-trivial tensor-stable positive map $\Pm:\M_{d_1}\ra\M_{d_2}$, then there exists $p\in\lbr -1,0\rb$ such that the map \eqref{equ:OneParamFamMap} is tensor-stable positive.

\item[(ii)] If for some $p\in[-1,0)$ the map \eqref{equ:OneParamFamMap} is tensor-stable positive, then it is non-trivial tensor-stable positive (i.e. it is neither completely positive nor completely co-positive).
\end{itemize}
\label{thm:OneParamFambl}
\end{thm} 
 
The aforementioned connection to the theory of entanglement distillation has the following direct implication:

\begin{thm}[Non-trivial tensor-stable positivity implies NPPT-bound entanglement]

If there exists a non-trivial tensor-stable positive map $\Pm:\M_{d_1}\ra\M_{d_2}$, then there exist NPPT bound-entangled states~\cite{PhysRevLett.80.5239,PhysRevA.61.062313,PhysRevA.61.062312} in $\M_{d_1}\otimes\M_{d_1}$ as well as in $\M_{d_2}\otimes\M_{d_2}$. 
\label{thm:NPTImpl}
\end{thm}

After completion of this work, we learned that tensor-stable positive maps have been introduced by M.\ Hayashi under the name ``tensor product positive maps'' in \cite[chapter 5]{hayashiQIbook}, where it was furthermore shown that the quantum relative entropy does not increase under the application of any trace-preserving tensor product positive map.

\section{Notation and preliminaries} 
\label{sec:Notation}

For every $d\in\N$, we fix an orthonormal basis $\lset\ket{i}\rset^d_{i=1}$ of the Hilbert space $\C^d$, and  denote by $\vartheta_d(X):=X^T$ the transposition w.r.t. that basis, the $d$-dimensional maximally entangled state by $\ket{\Omega_d} := \frac{1}{\sqrt{d}}\sum^d_{i=1}\ket{ii}\in\M_{d^2}$ and the corresponding projection by $\omega_d := \proj{\Omega_d}{\Omega_d}$. The $d\times d-$identity matrix will be denoted by $\id_d$. The following Lemma collects two frequently used and well-known techniques involving the maximally entangled state and linear maps that can be proved by direct computation.  

\begin{lem}[Tricks using the maximally entangled state]\hfill
\begin{enumerate}
\item For any $d_2\times d_1$-matrix $X$ we have $\lb \id_{d_1}\otimes X\rb\ket{\Omega_{d_1}} = \sqrt{\frac{d_2}{d_1}}\lb X^T\otimes \id_{d_2} \rb\ket{\Omega_{d_2}}$.
\item For any map $\Lm:\M_{d_1}\ra\M_{d_2}$ that is hermiticity-preserving (i.e. maps hermitian matrices to hermitian matrices), we have $\lb\ident_{d_1}\otimes \Lm\rb\lb \omega_{d_1}\rb = \frac{d_2}{d_1}\lb \vartheta_{d_1}\circ\Lm^*\circ\vartheta_{d_2}\otimes \ident_{d_2}\rb\lb \omega_{d_2}\rb$. 
\end{enumerate}
In the above $\Lm^*$ denotes the adjoint w.r.t. the Hilbert-Schmidt inner product.
\label{Lemma:tricks}
\end{lem}

We will frequently make use of the Choi-Jamiolkowski isomorphism between linear maps $\Lm:\M_{d_1}\ra\M_{d_2}$ and matrices $C\in\M_{d_1}\otimes \M_ {d_2}$. The Choi matrix of such a linear map is defined as $C_\Lm := \lb\ident_{d_1}\otimes \Lm\rb\lb\omega_{d_1}\rb$. Note that we used the normalized maximally entangled state in this definition. The following implications are well known:
\begin{itemize}
\item $\Lm:\M_{d_1}\ra\M_{d_2}$ is positive iff $C_\Lm$ is block-positive, i.e. $\lb\bra{\phi}\otimes \bra{\psi}\rb C\lb\ket{\phi}\otimes \ket{\psi}\rb \geq 0$ for all $\ket{\phi}\in\C^{d_1}$, $\ket{\psi}\in\C^{d_2}$.
\item $\Lm:\M_{d_1}\ra\M_{d_2}$ is completely positive iff $C_\Lm\geq 0$.
\item $\Lm:\M_{d_1}\ra\M_{d_2}$ is completely co-positive iff $C^{T_2}_\Lm\geq 0$.
\end{itemize}
For $C\in\M_{d_1}\otimes \M_{d_2}$ we denote by $C^{T_2} := \lb\ident_d\otimes \vartheta_d\rb\lb C\rb$ the partial transpose w.r.t. to the second tensor-factor. The paradigm of a block-positive matrix that is not positive is the Choi matrix of the transposition $\omega_d^{T_2} = \frac{1}{d}\mathbb{F}_d$. Here $\mathbb{F}_d:\C^d\otimes\C^d\ra \C^d\otimes\C^d$ denotes the flip operator with $\mathbb{F}_d\ket{ij} = \ket{ji}$.

Matrices $C\in\M_{d_1}\otimes \M_{d_2}$ with $C^{T_2}\geq 0$ will be called PPT (positive partial transpose). A matrix is called NPPT (non-positive partial transpose) if it is not PPT. The question of NPPT-bound entanglement~\cite{PhysRevLett.80.5239,PhysRevA.61.062313,PhysRevA.61.062312} concerns the problem of creating a maximally entangled state from many copies of an NPPT-state using only local operations and classical communications (LOCC)~\cite{chitambar2014everything}. While it is clear that no maximally entangled state can be created from many copies of a PPT-state it is currently unknown whether the same can be true for an NPPT-state.

For a linear map $\Lm:\M_{d_1}\ra\M_{d_2}$ we define the $\diamond$-norm~\cite{paulsen2002completely} as $\|\Lm\|_{\diamond} := \sup_{n\in\N}\|\ident_{n}\otimes\Lm\|_{1\ra 1}$. Here $\|\Sm\|_{1\ra 1}:=\sup_{\| X\|_1 = 1}\|\Sm\lb X\rb\|_1$ denotes the $1\ra 1$-norm of a linear map $\Sm$. By duality we have $\|\Sm\|_{1\ra 1}=\|\Sm^*\|_{\infty\ra\infty}=\sup_{\| X\|_\infty = 1}\|\Sm^*\lb X\rb\|_\infty$, where $\Sm^*$ denotes the Hilbert-Schmidt adjoint of $\Sm$. In the following lemma we collect some well-known properties of the $\diamond$-norm. 
\begin{lem}[Properties of the $\diamond$-norm]\hfill
\begin{enumerate}
\item For any completely positive map $\Tm:\M_{d_1}\ra\M_{d_2}$ we have
\begin{align}
\|\Tm\|_\diamond = \|\Tm\|_{1\ra 1} = \|\Tm^{*}\|_{\infty\ra\infty} = \|\Tm^{*}\lb\one_{d_2}\rb\|_\infty.
\label{equ:blablup}
\end{align}
\item For any linear map $\Lm:\M_{d_1}\ra\M_{d_2}$ we have
\begin{align*}
\|\Lm^{\otimes n}\|_\diamond = \|\Lm\|^n_\diamond.
\end{align*}
\item For any linear maps $\Lm_1:\M_{d_1}\ra\M_{d_2}$ and $\Lm_2:\M_{d_2}\ra\M_{d_3}$ we have
\begin{align*}
\|\Lm_2\circ \Lm_1\|_\diamond\leq \|\Lm_2\|_\diamond\|\Lm_1\|_\diamond.
\end{align*}
\end{enumerate}
\label{lem:PropDiam}
\end{lem}

See~\cite{paulsen2002completely} for proofs of these statements.

\section{Proof of Theorem \ref{ntensor}}
\label{sec:FamiliesN}

Our proof of Theorem \ref{ntensor} uses the following quantitative version of the result \cite[Lemma 22]{6094278} about tensor products of generalizations of unextendible product bases \cite{PhysRevLett.82.5385}, whose elements are not necessarily mutually orthogonal. For the following, we call a matrix $P\in\M_{d_1}\otimes\M_{d_2}$ \emph{separable} if it can be written as $P=\sum^k_{i=1} A_i\otimes B_i$ for some $k\in\N$ and matrices $A_i\in\M_{d_1}$, $B_i\in\M_{d_2}$ with $A_i\geq 0$ and $B_i\geq 0$. 

\begin{lem}[Multiplicativity of minimal overlap with product states]\label{cubittlemma}

For a separable matrix $P\in\M_{d_1}\otimes\M_{d_2}$, define
\begin{align*}
\mu ~:= ~\min\lset\lb\bra{\psi}\otimes\bra{\phi}\rb P \lb\ket{\psi}\otimes\ket{\phi}\rb : \ket{\psi}\in\C^{d_1},\ket{\phi}\in\C^{d_2},\braket{\psi}{\psi} = \braket{\phi}{\phi}=1   \rset .
\end{align*}
Then, for all $n\in \N$, we have
\begin{align*}
\min\lset\lb\bra{\Psi}\otimes\bra{\Phi}\rb P^{\otimes n} \lb\ket{\Psi}\otimes\ket{\Phi}\rb:\ket{\Psi}\in\lb\C^{d_1}\rb^{\otimes n},\ket{\Phi}\in\lb\C^{d_2}\rb^{\otimes n}, \braket{\Psi}{\Psi} = \braket{\Phi}{\Phi}=1  \rset ~=~ \mu^n .
\end{align*}
In particular, if there is no nonzero product vector in the kernel of $P$, then there is none in the kernel of $P^{\otimes n}$.
\end{lem}

The connection to \cite[Lemma 22]{6094278} becomes clear by noting that any separable matrix $P\in\M_{d_1}\otimes\M_{d_2}$ admits a decomposition of the form $P\,=\,\sum_{i=1}^N\ket{\psi_i}\bra{\psi_i}\otimes\ket{\phi_i}\bra{\phi_i}$ such that $\text{ker}\lb P\rb = \lb\text{span}\lset \ket{\psi_i}\otimes \ket{\phi_i}\rset^N_{i=1}\rb^\perp$. Hence for $\mu>0$ the set $\lset \ket{\psi_i}\otimes \ket{\phi_i}\rset$ forms an unextendible product set.

For the following proof we will need the \emph{minimal output eigenvalue} of a completely positive map $\Tm:\M_{d_1}\ra\M_{d_2}$ defined as
\begin{align}
\lambda^{\text{min}}_{\text{out}}\lbr \Tm\rbr := \min_{\rho\in\D_{d_1}} \lambda_{\text{min}}\lb \Tm(\rho)\rb\,.
\label{minoutev}
\end{align}
Here $\lambda_{\text{min}}\lb \cdot\rb$ denotes the minimal eigenvalue and $\D_{d_1}$ is the set of quantum states in $\M_{d_1}$. For any entanglement breaking map $\Tm$ and any completely positive map $\Sm$ we prove in Theorem \ref{Additivity} (Appendix \ref{Appendix2}) that 
\begin{align*}
\lambda^{\text{min}}_{\text{out}}\lbr \Tm\otimes \Sm\rbr = \lambda^{\text{min}}_{\text{out}}\lbr \Tm\rbr \lambda^{\text{min}}_{\text{out}}\lbr \Sm\rbr.
\end{align*}
Thus, $\lambda^{\text{min}}_{\text{out}}$ is multiplicative for entanglement breaking maps.

\begin{proof}[proof of Lemma \ref{cubittlemma}]

Consider the completely positive map $\Tm:\M_{d_1}\ra\M_{d_2}$ such that $P = C_{\Tm}$. Then we have
\begin{align*}
\lb\bra{\Psi}\otimes\bra{\Phi}\rb P^{\otimes k} \lb\ket{\Psi}\otimes\ket{\Phi}\rb = \frac{1}{d_1^k}\bra{\Phi} \Tm^{\otimes k}\lb\overline{\proj{\Psi}{\Psi}}\rb \ket{\Phi}
\end{align*}
for all $k\in\N$ and all $\ket{\Psi}\in\lb\C^{d_1}\rb^{\otimes k}$, $\ket{\Phi}\in\lb\C^{d_2}\rb^{\otimes k}$. Using the minimal output eigenvalue \eqref{minoutev} we have for any $k\in\N$
\begin{align*}
\lambda^{\text{min}}_{\text{out}}( \Tm^{\otimes k}) =d_1^k\min\big\{\lb\bra{\Psi}\bra{\Phi}\rb P^{\otimes k} \lb\ket{\Psi}\ket{\Phi}\rb:\ket{\Psi}\in(\C^{d_1})^{\otimes k},\ket{\Phi}\in(\C^{d_2})^{\otimes k}, \|\Psi\| = \|\Phi\|=1  \big\}.
\end{align*}
As $P$ is separable the map $\Tm$ is entanglement breaking~\cite{horodecki2003entanglement} and we can apply Theorem \ref{Additivity} from Appendix \ref{Appendix2}. This shows that $\lambda^{\text{min}}_{\text{out}}\lb \Tm^{\otimes n}\rb = \lambda^{\text{min}}_{\text{out}}\lb \Tm\rb^n$ and finishes the proof.

\end{proof}

With this ingredient we prove Theorem \ref{ntensor}:
\begin{proof}[proof of Theorem \ref{ntensor}]Choose orthonormal bases $\{\ket{i}\}\subseteq\C^{d_1}$ and $\{\ket{j}\}\subseteq\C^{d_2}$ and define the operator
\begin{align}\nonumber
P~:=~&\big(\ket{1}\ket{1}+\ket{2}\ket{2}\big)\big(\bra{1}\bra{1}+\bra{2}\bra{2}\big)\,+\,\ket{1}\ket{2}\bra{1}\bra{2}\,+\,\ket{2}\ket{1}\bra{2}\bra{1}\\
&~~+\sum_{\substack{{(i,j)}\\{i>2~\text{or}~j>2}}}\ket{i}\ket{j}\bra{i}\bra{j}~~\in~\M_{d_1}\otimes\M_{d_2}\,.\label{Pforunextendible}
\end{align}
It is easy to verify that 
\begin{align*}
P\, = \, \sum^3_{k=1} \frac{1}{3}\proj{\xi_k}{\xi_k}\otimes\overline{\proj{\xi_k}{\xi_k}}+\sum_{\substack{{(i,j)}\\{i>2~\text{or}~j>2}}}\proj{i}{i}\otimes\proj{j}{j}
\end{align*}
for $\ket{\xi_k} = \ket{1} + e^{\frac{2\pi ik}{3}}\ket{2}$. This shows that $P$ is separable as a sum of positive product operators, and for later we note $\|P\|_\infty=2$. Now define
\begin{align*}
\mu ~:= ~\min\lset\lb\bra{\psi}\otimes\bra{\phi}\rb P \lb\ket{\psi}\otimes\ket{\phi}\rb : \ket{\psi}\in\C^{d_1},\ket{\phi}\in\C^{d_2}, \braket{\psi}{\psi} = \braket{\phi}{\phi}=1   \rset 
\end{align*}
and apply Lemma \ref{cubittlemma} showing that for any $k\in\N$: 
\begin{align*}
\min\lset\lb\bra{\Psi}\otimes\bra{\Phi}\rb P^{\otimes k} \lb\ket{\Psi}\otimes\ket{\Phi}\rb:\ket{\Psi}\in\lb\C^{d_1}\rb^{\otimes k},\ket{\Phi}\in\lb\C^{d_2}\rb^{\otimes k}, \braket{\Psi}{\Psi} = \braket{\Phi}{\Phi}=1  \rset ~=~ \mu^k .
\end{align*}
As the kernel ${\rm ker}(P)={\rm span}\{\ket{1}\ket{1}-\ket{2}\ket{2}\}$ of $P$ in \eqref{Pforunextendible} contains no nonzero product vector we have $\mu > 0$. One can actually compute $\mu=1/2$. With this we can compute 
\begin{align}
\label{equ:ntsBound}
\big(\bra{\Psi}\otimes\bra{\Phi}\big)\,(P-\varepsilon\id_{d_1}\id_{d_2})^{\otimes n}\,\big(\ket{\Psi}&\otimes\ket{\Phi}\big)\,\geq \sum^{\lfloor\frac{n}{2}\rfloor}_{k=0} \binom{n}{2k}\varepsilon^{2k} \mu^{n-2k} - \sum^{\lfloor\frac{n+1}{2}\rfloor}_{k=1}\binom{n}{2k-1}\varepsilon^{2k-1} \| P\|_\infty^{n-2k+1} \\
& = \frac{(\mu + \varepsilon)^n + (\mu - \varepsilon)^n}{2} - \frac{(\| P\|_\infty + \varepsilon)^n - (\| P\|_\infty - \varepsilon)^n}{2} \nonumber\\
&\geq \mu^n - (\| P\|_\infty + \varepsilon)^n + \| P\|^n_\infty\geq 0\nonumber
\end{align}
for any $0\leq\varepsilon \leq  \sqrt[n]{\| P\|^n_\infty + \mu^n} - \| P\|_\infty$. This means that $(P-\varepsilon\id_{d_1d_2})^{\otimes n}\in(\M_{d_1})^{\otimes n}\otimes(\M_{d_2})^{\otimes n}$ is a block-positive operator for any $\varepsilon\in \lbr 0, \sqrt[n]{\| P\|^n_\infty + \mu^n} - \| P\|_\infty\rbr$, which by the Choi-Jamiolkowski isomorphism (Section \ref{sec:Notation}) corresponds to a positive linear map $\Pm_\varepsilon^{\otimes n}:(\M_{d_1})^{\otimes n}\ra(\M_{d_2})^{\otimes n}$. The map $\Pm_\varepsilon :\M_{d_1}\ra\M_{d_2}$ with Choi matrix $(P-\varepsilon\id_{d_1d_2})$ is thus $n$-tensor-stable positive. Note that $P$ is rank-deficient and as $P^{T_2}$ equals the expression (\ref{Pforunextendible}) with the first terms replaced by $(\ket{1}\ket{2}+\ket{2}\ket{1})(\bra{1}\bra{2}+\bra{2}\bra{1})+\ket{1}\ket{1}\bra{1}\bra{1}+\ket{2}\ket{2}\bra{2}\bra{2}$ it is rank-deficient as well. Hence, the Choi matrices $(P-\varepsilon\id_{d_1d_2})$ and $(P^{T_2}-\varepsilon\id_{d_1d_2})$ of $\Pm$ respectively $\vartheta_{d_2}\circ\Pm$ are not positive for $\epsilon >0$, which finally shows that $\Pm_\varepsilon$ is not a trivial tensor-stable positive map for any $\varepsilon\in \lb 0, \sqrt[n]{\| P\|^n_\infty + \mu^n} - \| P\|_\infty\rbr$, i.e. in particular for $\varepsilon \in (0,\frac{2}{8^n}]$.
\end{proof}

\section{Applications to quantum information theory}
\label{sec:Appl}

Deciding the existence of non-trivial tensor-stable positive maps could lead to a solution of other open problems in quantum information theory. Here we will discuss two such connections. 

\subsection{Entanglement annihilating channels}

In \cite{moravvcikova2010entanglement,filippov2012local,filippov2013bipartite,filippov2013dissociation} the authors study how entanglement in a multipartite setting can be destroyed by dissipative processes. They define the set of $k$-locally entanglement annihilating channels. These are quantum channels $\Tm:\M_{d_1}\ra\M_{d_2}$ such that $\Tm^{\otimes k}\lb\rho\rb$ is $k$-partite separable for all input states $\rho\in\M_{d_1^k}$, i.e. for all $\rho\geq 0$ we have $\Tm^{\otimes k}\lb\rho\rb = \sum^m_{i=1} p_i \sigma^{(1)}_i\otimes \sigma^{(2)}_i\otimes\cdots \otimes \sigma^{(k)}_i$ for some $m\in\N$, states $\sigma^{(j)}_i \in\M_{d_2}$ and $p_i\in\R^+$ depending on $\rho$. Furthermore, a channel is called $\infty$-locally entanglement annihilating if it is $k$-locally entanglement annihilating for all $k\in\N$.

It is clear that entanglement breaking channels~\cite{horodecki2003entanglement} are $\infty$-locally entanglement annihilating. In \cite{moravvcikova2010entanglement,filippov2012local} examples of $2$-locally entanglement annihilating channels are constructed that are not entanglement breaking. However it is not known whether there exist an $\infty$-locally entanglement annihilating channel, which is not entanglement breaking.

We can prove the following theorem connecting $k$-locally entanglement annihilating channels to tensor-stable positive maps. 

\begin{thm}

If the quantum channel $\Tm:\M_{d_1}\ra\M_{d_2}$ is $k$-locally entanglement annihilating for some $k\geq 2$, but not entanglement breaking, then there exists a positive map $\Sm:\M_{d_2}\ra\M_{d_1}$ such that $\Pm:\M_{d^{2}_1}\ra\M_{d^2_1}$ defined as 
\begin{align}
\Pm = \lb\Sm\circ\Tm\rb\otimes \lb\vartheta\circ\Sm\circ\Tm\rb
\label{equ:mapPosFromEA}
\end{align}
is a $\lfloor\frac{k}{2}\rfloor$-tensor-stable positive map that is \emph{not} a trivial tensor-stable positive map.

Thus, the existence of a non-entanglement breaking $\infty$-locally entanglement annihilating channel implies the existence of a non-trivial tensor-stable positive map.

\end{thm}

\begin{proof}

Assume that $\Tm:\M_{d_1}\ra\M_{d_2}$ is a $k$-locally entanglement annihilating channel. If $\Tm$ is not entanglement breaking, then there exists a positive map $\Sm:\M_{d_2}\ra\M_{d_1}$ such that $\Sm\circ\Tm$ is not completely positive~\cite{horodecki1996separability}. Now consider the map $\Pm:\M_{d^{2}_1}\ra\M_{d^2_1}$ defined in \eqref{equ:mapPosFromEA}.
As $\Tm:\M_{d_1}\ra\M_{d_2}$ is $k$-entanglement annihilating, $\Pm$ is $\lfloor \frac{k}{2}\rfloor$-tensor-stable positive. Furthermore it is neither completely positive nor completely co-positive.

\end{proof}

By our Theorem \ref{thm:NPTImpl}, the existence of a $\infty$-locally entanglement annihilating but not entanglement breaking channel then implies the existence of NPPT-bound entanglement.  

\subsection{Upper bounds on the quantum capacity}\label{capacityboundsubsection}

The existence of non-trivial tensor-stable positive maps would imply new bounds on the quantum capacity of a quantum channel. By generalizing the proof of the transposition criterion~\cite{holevo2001evaluating,kretschmann2004tema} we obtain a quantitative bound on the quantum capacity $\mathcal{Q}\lb\Tm\rb$ of a quantum channel. Recall that the quantum capacity is defined as:

\begin{defn}[\text{Quantum capacity }$\mathcal{Q}$, see ~\cite{PhysRevA.55.1613, kretschmann2004tema}]\hfill
\label{defn:QuantCap}

The \textbf{quantum capacity} of a quantum channel $\Tm:\M_{d_1}\ra \M_{d_2}$ is defined as 
\begin{align*}
\mathcal{Q}\lb \Tm\rb := \sup\lset R\in \R^+ :\text{ R achievable rate}\rset ,
\end{align*}
where a rate $R\in\R^+$ is called achievable if there exist sequences $\lb n_\nu\rb^\infty_{\nu=1},\lb m_\nu\rb^\infty_{\nu=1}$ such that $R = \limsup_{\nu\ra \infty}\frac{n_\nu\log_2(d)}{m_\nu}$ and the approximation error vanishes in the asymptotic limit, i.e. 
\begin{align}
\inf_{\Em,\Dm}\,\frac{1}{2}\left\Vert \ident_{d}^{\otimes n_\nu} - \Dm\circ \Tm^{\otimes m_\nu}\circ \Em\right\Vert_\diamond \ra 0 \hspace*{0.3cm}\text{as } \nu\ra \infty .
\label{equ:ApproxError}
\end{align}
Here, the infimum runs over all encoding and decoding quantum channels $\Em:\M^{\otimes n_\nu}_{d}\ra \M^{\otimes m_\nu}_{d_1}$ and $\Dm:\M^{\otimes m_\nu}_{d_2}\ra \M^{\otimes n_\nu}_{d}$, and $d\geq2$ is any fixed integer (note, the value of $\mathcal{Q}\lb \Tm\rb$ does not depend on the choice of $d$ \cite{kretschmann2004tema}).
\end{defn} 

Currently all channels known to have zero quantum capacity come from two classes~\cite{PhysRevLett.108.230507}. These are the classes of anti-degradable channels~\cite{bruss1998optimal,bennett1997capacities} and of completely co-positive quantum channels. The latter can be shown using the quantitative transposition bound \cite{holevo2001evaluating}
\begin{align}
\mathcal{Q}\lb\Tm\rb\leq \log_2\lb \norm{\vartheta_{d_2}\circ\Tm}{\diamond}\rb
\label{equ:transBound}
\end{align}
on the quantum capacity of any quantum channel $\Tm:\M_{d_1}\ra\M_{d_2}$. We will now prove a generalization of this bound using any surjective, unital and tensor-stable positive map $\Pm:\M_{d_3}\to\M_{d_2}$ that is not completely positive. Note that any surjective linear map $\Pm:\M_{d_3}\to\M_{d_2}$ has a linear right-inverse $\Pm^{-1}:\M_{d_2}\to\M_{d_3}$ (generally not unique) satisfying $\Pm\circ\Pm^{-1}=\ident_{d_2}$.


\begin{thm}
\label{thm:CapBound}
Let $\Tm:\M_{d_1}\ra\M_{d_2}$ be a quantum channel and $\Pm:\M_{d_3}\ra\M_{d_2}$ be a surjective, unital and tensor-stable positive map that is not completely positive, and let $\Pm^{-1}$ be any right-inverse of $\Pm$. Then we have
\begin{align*}
\mathcal{Q}\lb\Tm\rb\leq \frac{\log_2\lb\norm{\Pm^{-1}\circ\Tm}{\diamond}\norm{\Pm^{*}\lb\id_{d_2}\rb}{\infty}\rb\log_2(d_2)}{\log_2\lb\norm{\Pm^{*}}{\diamond}\rb}
\end{align*}

Note that the transposition bound \eqref{equ:transBound} is retrieved for $\Pm = \vartheta_{d_2}$.
\end{thm} 

\begin{proof}

As $\Pm^*$ is trace-preserving but not completely positive, we have $\norm{\Pm^{*}}{\diamond}>1$~\cite{paulsen2002completely}. Furthermore note that for any $n\in\N$ and any $X\in\M_n\otimes \M_{d_2}$ we have 
\begin{align*}
\|\lb\ident_n\otimes \vartheta_{d_3}\circ\Pm^*\circ\vartheta_{d_2}\rb\lb X\rb\|_1 &= \|\lbr\lb\ident_n\otimes \vartheta_{d_3}\circ\Pm^*\circ\vartheta_{d_2}\rb\lb X\rb\rbr^T\|_1\\
&=\|\lb\ident_n\otimes \Pm^*\rb\lb X^T\rb\|_1
\end{align*}
as the transposition does not change the spectrum. By the definition of the diamond norm this implies $\norm{\Pm^{*}}{\diamond}=\norm{\vartheta_{d_3}\circ\Pm^{*}\circ\vartheta_{d_2}}{\diamond}$.

Now we can do the following calculation, which generalizes the proof of the transposition bound~\cite{holevo2001evaluating,kretschmann2004tema}. Let $\Em:\M^{\otimes n_\nu}_{d_2}\ra \M^{\otimes m_\nu}_{d_1}$ and $\Dm:\M^{\otimes m_\nu}_{d_2}\ra \M^{\otimes n_\nu}_{d_2}$ denote arbitrary quantum channels. Then:
\begin{align*}
\norm{\vartheta_{d_3}&\circ\Pm^*\circ \vartheta_{d_2}}{\diamond}^{n_\nu} = \norm{(\vartheta_{d_3}\circ\Pm^*\circ \vartheta_{d_2})^{\otimes n_\nu}\circ \lb \ident^{\otimes n_\nu}_{d_2} - \Dm\circ\Tm^{\otimes m_\nu}\circ\Em + \Dm\circ\Tm^{\otimes m_\nu}\circ\Em\rb}{\diamond} \\
 &\leq \norm{\lb\vartheta_{d_3}\circ\Pm^*\circ \vartheta_{d_2}\rb^{\otimes n_\nu}\circ( \ident^{\otimes n_\nu}_{d_2} - \Dm\circ\Tm^{\otimes m_\nu}\circ\Em)}{\diamond} + \norm{\lb\vartheta_{d_3}\circ\Pm^*\circ \vartheta_{d_2}\rb^{\otimes n_\nu}\circ\Dm\circ\Tm^{\otimes m_\nu}\circ\Em}{\diamond} \\
&\leq 2\epsilon_\nu\norm{\vartheta_{d_3}\circ\Pm^*\circ \vartheta_{d_2}}{\diamond}^{n_\nu} + \norm{\lb\vartheta_{d_3}\circ\Pm^*\circ \vartheta_{d_2}\rb^{\otimes n_\nu}\circ\Dm\circ \Pm^{\otimes m_\nu}}{\diamond}\,\norm{\Pm^{-1}\circ\Tm}{\diamond}^{m_\nu},
\end{align*}
with $\epsilon_\nu := \norm{\ident^{\otimes n_\nu}_{d_2} - \Dm\circ\Tm^{\otimes m_\nu}\circ\Em}{\diamond}/2$. Here we used the triangle inequality for the first inequality and the properties from Lemma \ref{lem:PropDiam} for the second inequality (in particular we used $\|\Em\|_\diamond = 1$). Note that by Lemma \ref{Lemma:tricks} we have
\begin{align*}
\lb\ident_{d_3^{m_\nu}}\otimes\lbr\lb\vartheta_{d_3}\circ\Pm^*\circ \vartheta_{d_2}\rb^{\otimes n_\nu}\circ\Dm\circ \Pm^{\otimes m_\nu}\rbr\rb&\lb\omega^{\otimes m_\nu}_{d_3}\rb\\
= \left(\frac{d_2}{d_3}\right)^{m_\nu}\vartheta_{d_3^{\otimes\lb m_\nu + n_\nu\rb}}\circ\lb\Pm^{*}\rb^{\otimes\lb m_\nu + n_\nu\rb}\circ&\,\vartheta_{d_2^{\otimes\lb m_\nu + n_\nu\rb}}\circ \lb\ident_{d_2^{m_\nu}}\otimes \Dm\rb\lb\omega^{\otimes m_\nu}_{d_2}\rb~\geq~0,
\end{align*}
since $\Pm^*$ is also tensor-stable positive. Thus, the map $\lb\vartheta_{d_3}\circ\Pm^*\circ \vartheta_{d_2}\rb^{\otimes n_\nu}\circ\Dm\circ \Pm^{\otimes m_\nu}$ is completely positive. Therefore, we can apply Lemma \ref{lem:PropDiam} (equation \eqref{equ:blablup}) and obtain
\begin{align*}
\norm{\lb\vartheta_{d_3}\circ\Pm^*\circ \vartheta_{d_2}\rb^{\otimes n_\nu}\circ\Dm\circ \Pm^{\otimes m_\nu}}{\diamond}=\norm{(\Pm^*)^{\otimes m_\nu}\circ\Dm^*\circ \lb\vartheta_{d_2}\circ\Pm\circ \vartheta_{d_3}\rb^{\otimes n_\nu}\lb\one^{\otimes n_\nu}_{d_3}\rb}{\infty} = \norm{\Pm^*\lb\id_{d_2}\rb}{\infty}^{m_\nu}
\end{align*} 
for all quantum channels $\Dm$, where we used unitality of $\Pm$ and that $\Dm$ is trace-preserving. Inserting this into the above calculation we have
\begin{align*}
\lb 1-2\epsilon_\nu\rb\norm{\Pm^*}{\diamond}^{n_\nu} = \lb 1-2\epsilon_\nu\rb\norm{\vartheta_{d_3}\circ\Pm^*\circ \vartheta_{d_2}}{\diamond}^{n_\nu}\leq \norm{\Pm^{*}\lb\id_{d_2}\rb}{\infty}^{m_\nu}\norm{\Pm^{-1}\circ\Tm}{\diamond}^{m_\nu}. 
\end{align*}
Applying the logarithm and taking the limit $\nu\ra\infty$ we obtain
\begin{align*}
R = \limsup_{\nu\ra\infty}\frac{n_\nu\log_2(d_2)}{m_\nu}\leq \frac{\log_2\lb\norm{\Pm^{-1}\circ\Tm}{\diamond}\norm{\Pm^{*}\lb\id_{d_2}\rb}{\infty}\rb\log_2(d_2)}{\log_2\lb\norm{\Pm^{*}}{\diamond}\rb}
\end{align*}
for any achievable rate $R$ (see Definition \ref{defn:QuantCap}) and corresponding coding schemes $\Em,\Dm$ with $\epsilon_\nu\ra 0$. 
\end{proof}


To apply Theorem \ref{thm:CapBound} it is enough to have a surjective and tensor-stable positive map $\Rm:\M_{d_3}\to\M_{d_2}$ which is not completely positive. Note that as $\Rm$ is surjective, it is easy to see that the operator $\Rm(\id_{d_3})$ is strictly positive, and thus the map $\Pm:\M_{d_3}\to\M_{d_2}$ defined by $\Pm(X):=\Rm(\id_{d_3})^{-1/2}\Rm(X)\Rm(\id_{d_3})^{-1/2}$ is unital, surjective and tensor-stable positive. Furthermore, $\Pm$ is completely (co-)positive if and only if $\Rm$ was completely (co-)positive. Thus, we constructed a map $\Pm$ as needed for Theorem \ref{thm:CapBound}.

Note that for completely co-positive maps $\Pm$ the capacity bound from Theorem \ref{thm:CapBound} is worse than the transposition bound given by \eqref{equ:transBound}. To prove this let $\Pm = \vartheta_{d_2}\circ \Sm$ for a surjective, unital and completely positive map $\Sm:\M_{d_3}\ra\M_{d_2}$. Then, due to the invertibility of $\vartheta_{d_2}$, any right-inverse $\Pm^{-1}$ of $\Pm$ can be written as $\Pm^{-1}=\Sm^{-1}\circ\vartheta_{d_2}$ with a right-inverse $\Sm^{-1}:\M_{d_2}\to\M_{d_3}$ of $\Sm$.  By unitality of $\Pm$ and basic properties of the $\diamond$-norm (see for instance~\cite[Exercise 3.11 and Corollary 2.9]{paulsen2002completely}) we have $\norm{\Pm^*}{\diamond}\leq d_2\norm{\Pm^*}{1\ra 1} = d_2$, and furthermore $\norm{\Pm^*\lb \id_{d_2}\rb}{\infty} = \norm{\Sm^*\lb \id_{d_2}\rb}{\infty} = \norm{\Sm}{\diamond}$ since $\Sm$ is completely positive. Thus, for any quantum channel $\Tm:\M_{d_1}\ra\M_{d_2}$ we have:
\begin{align*}
\frac{\log_2\lb\norm{\Pm^{-1}\circ\Tm}{\diamond}\norm{\Pm^{*}\lb\id_{d_2}\rb}{\infty}\rb\log_2d_2}{\log_2\lb\norm{\Pm^{*}}{\diamond}\rb} \geq \log_2\lb \norm{\Sm^{-1}\circ\vartheta_{d_2}\circ\Tm}{\diamond}\norm{\Sm}{\diamond}\rb \geq \log_2 \norm{\vartheta_{d_2}\circ\Tm}{\diamond}\geq \mathcal{Q}\lb\Tm\rb.
\end{align*} 
Therefore, to obtain a capacity bound stronger than the transposition bound \eqref{equ:transBound}, one would need a non-trivial tensor-stable positive map $\Pm$.  

Similarly, if $\Pm:\M_{d_1}\ra\M_{d_3}$ is a trace-preserving and tensor-stable positive map that is not completely positive and that has a left-inverse $\Pm^{-1}:\M_{d_3}\to\M_{d_1}$, then the following bound holds for any quantum channel $\Tm:\M_{d_1}\ra\M_{d_2}$:
\begin{align}\label{otherversionQbound}
\mathcal{Q}\lb\Tm\rb\leq \frac{\log_2\lb\norm{\Tm\circ\Pm^{-1}}{\diamond}\rb\log_2(d_1)}{\log_2\lb\norm{\Pm^{*}}{\diamond}/\norm{\Pm\lb\id_{d_1}\rb}{\infty}\rb}.
\end{align}
The proof works in the same way as the proof of Theorem \ref{thm:CapBound}, and again, this bound reduces to the transposition bound (\ref{equ:transBound}) for $\Pm=\vartheta_{d_1}$.

%
%
%
%
%
%

\subsection{Transposition bound as a strong converse rate for the two-way quantum capacity}
\label{twowaysubsection}

We now prove that the transposition bound (\ref{equ:transBound}) is even an upper bound on the capacity ${\mathcal{Q}}_2(\Tm)\geq{\mathcal{Q}}(\Tm)$ of any channel $\Tm$ for forward communication of quantum information assisted by unrestricted two-way classical side communication between both parties and arbitrary local quantum operations (LOCC).

For this, we first define an \emph{LOCC channel} (w.r.t.\ bipartitions $A:B$ and $A':B'$ of the input and output systems, respectively) to be any quantum channel $\Lm_{A:B\to A':B'}:\M_{d_A}\otimes\M_{d_{B}}\to\M_{d_{A'}}\otimes\M_{d_{B'}}$ that can be written as a sequential concatenation of any number of channels $\Lm_{A_q:B_q\to A'_qA'_c:B'_qB'_c}$ of the following form ($X_{A_qB_q}\in\M_{d_{A_q}}\otimes\M_{d_{B_q}}$):
\begin{align}\label{LOCCkraus}
\Lm_{A_q:B_q\to A'_qA'_c:B'_qB'_c}(X_{A_qB_q})=\sum_{i,j}(K^A_i\otimes K^B_j)X_{A_qB_q}(K^A_i\otimes K^B_j)^\dagger\otimes\ket{j}\bra{j}_{A'_c}\otimes\ket{i}\bra{i}_{B'_c},
\end{align}
where $K^A_i:\C^{|A_q|}\to\C^{|A'_q|}$ and $K^B_j:\C^{|B_q|}\to\C^{|B'_q|}$ $(i\in I, j\in J)$ are Kraus operators of quantum channels mapping system $A_q$ to $A'_q$ and system $B_q$ to $B'_q$ respectively (i.e.\ $\sum_i(K^A_i)^\dagger K^A_i=\id_{A_q}$ and $\sum_j(K^B_j)^\dagger K^B_j=\id_{B_q}$), and $\ket{j}_{A'_c}$ and $\ket{i}_{B'_c}$ are orthonormal bases belonging to (effectively classical) systems $A_c$ and $B_c$ of dimension $|J|$ and $|I|$ (see \cite{chitambar2014everything} for more details). When one of the systems, such as $B$, is trivial (i.e.\ one-dimensional), we also speak of a LOCC channel $\Lm_{A\to A':B'}$, omitting the indices of the trivial subsystems. From the definition it is clear that any LOCC channel $\Lm_{A:B\to A':B'}:\M_{d_A}\otimes\M_{d_{B}}\to\M_{d_{A'}}\otimes\M_{d_{B'}}$ is \emph{PPT preserving} (w.r.t.\ bipartitions $A:B$ and $A':B'$), meaning that the map $(\ident_{A'}\otimes\vartheta_{B'})\Lm_{A:B\to A':B'}(\ident_A\otimes\vartheta_B)$ is completely positive and therefore a quantum channel, whose $\diamond$-norm equals $1$. We can now define the two-way quantum capacity.

\begin{defn}[\text{Two-way quantum capacity }$\mathcal{Q}_2$]\label{defn:TwoWayCap}\hfill

Given a quantum channel $\Tm:\M_{d_1}\to\M_{d_2}$, we define an \emph{$(N,m,\varepsilon)$-scheme for quantum communication with two-way classical communication} to be any set of LOCC channels $\Lm_{A_i:B^t_iB_i\to A^t_{i+1}A_{i+1}:B_{i+1}}$ for $i=0,\ldots,m$, where the initial $A$-system and final $B$-system are of the same dimension $N=|A_0|=|B_{m+1}|$ and are identified with each other, $A_0=B_{m+1}$, the initial $B$-system and final $A$-system are trivial, $|B^t_0|=|B_0|=|A^t_{m+1}|=|A_{m+1}|=1$, and the subsystems used for quantum transmission (hence the superscript ``t'') are of dimensions $|A^t_i|=d_1$ and $|B^t_i|=d_2$ for $i=1,\ldots,m$, and $\varepsilon$ is the $\diamond$-norm error of the scheme,
\begin{align}\nonumber
\varepsilon=\big\|\ident_{A_0\to B_{m+1}}-&\Lm_{A_m:B^t_mB_m\to B_{m+1}}\circ\Tm_{A^t_m\to B^t_m}\circ\Lm_{A_{m-1}:B^t_{m-1}B_{m-1}\to A_mA^t_mB_m}\circ\Tm_{A^t_{m-1}\to B^t_{m-1}}\circ\ldots\\
&\ldots\,\circ\Tm_{A^t_2\to B^t_2}\circ\Lm_{A_1:B^t_1B_1\to A^t_2A_2:B_2}\circ\Tm_{A^t_1\to B^t_1}\circ\Lm_{A_0\to A^t_1A_1:B_1}\big\|_{\diamond}\,/\,2\,,\label{LOCCcodingscheme}
\end{align}
omitting for brevity the action of the identity channel on some subsystems, e.g.\ in $\Tm_{A^t_i\to B^t_i}\equiv(\Tm_{A^t_i\to B^t_i}\otimes\ident_{A_i}\otimes\ident_{B_i})$.

We call $R\in\R^+$ an \emph{achievable rate} for quantum communication over the channel $\Tm$ assisted by two-way classical communication if there exists for each $\nu\in\N$ a $(N_\nu,m_\nu,\varepsilon_\nu)$-scheme as just defined in such a way that $R=\limsup_{\nu\to\infty}\frac{\log_2(N_\nu)}{m_\nu}$ and $\lim_{\nu\to\infty}\varepsilon_\nu=0$. The \emph{two-way quantum capacity ${\mathcal{Q}}_2(\Tm)$} is defined to be the supremum of all such achievable rates.
\end{defn}

To prove the following statements about ${\mathcal{Q}}_2$ we need only the PPT preserving property of the LOCC channels in the above coding scheme. The statements hold therefore more generally for quantum communication assisted by any PPT preserving channels.

\begin{lem}[Error of two-way coding schemes]\label{TwoWayErrorLemma}
Let $\Tm:\M_{d_1}\to\M_{d_2}$ be a quantum channel and suppose there exists a $(N,m,\varepsilon)$-scheme for quantum communication with two-way classical side communication. Then:
\begin{align*}
\varepsilon\,\geq\,1-\frac{\big\|\vartheta_{d_2}\circ\Tm\big\|_\diamond^m}{N}\,.
\end{align*}
\end{lem}
\begin{proof}The following proof generalizes ideas from the examples in \cite[Section III]{morganwinter}. We follow through the $m$ steps of the given $(N,m,\varepsilon)$-scheme (cf.\ Definition \ref{defn:TwoWayCap}) and examine how the partially transposed communication channel between the two parties evolves. For this, let $\Sm^{(1)}_{A_0\to A^t_1A_1B_1}:=\Lm_{A_0\to A^t_1A_1B_1}$ and for $i=1,\ldots,m$,
\begin{align*}
\Sm^{(i+1)}_{A_0\to A^t_{i+1}A_{i+1}B_{i+1}}\,:=\,(\Lm_{A_i:B^t_iB_i\to A^t_{i+1}A_{i+1}:B_{i+1}})\circ(\Tm_{A^t_i\to B^t_i}\otimes\ident_{A_i}\otimes\ident_{B_i})\circ\Sm^{(i)}_{A_0\to A^t_iA_iB_i}\,.
\end{align*}
As each LOCC map in the communication scheme is PPT preserving and using that the transposition is an involution, i.e.\ $\vartheta_{B^t_iB_i}\circ(\vartheta_{B^t_i}\otimes\vartheta_{B_i})=\ident_{B^t_iB_i}$ we have:

\begin{align*}
\big\|(\ident_{A^t_{i+1}A_{i+1}}\otimes\vartheta_{B_{i+1}})&\circ\Sm^{(i+1)}_{A_0\to A^t_{i+1}A_{i+1}B_{i+1}}\big\|_{\diamond}\\
&=\big\|(\ident_{A^t_{i+1}A_{i+1}}\otimes\vartheta_{B_{i+1}})\circ(\Lm_{A_i:B^t_iB_i\to A^t_{i+1}A_{i+1}:B_{i+1}})\circ(\ident_{A_i}\otimes\vartheta_{B^t_iB_i})\,\circ \\
&\qquad\circ(\ident_{A_i}\otimes\vartheta_{B^t_i}\otimes\vartheta_{B_i})(\Tm_{A^t_i\to B^t_i}\otimes\ident_{A_i}\otimes\ident_{B_i})\circ\Sm^{(i)}_{A_0\to A^t_iA_iB_i}\big\|_\diamond\\
&\leq\big\|(\ident_{A^t_{i+1}A_{i+1}}\otimes\vartheta_{B_{i+1}})\circ(\Lm_{A_i:B^t_iB_i\to A^t_{i+1}A_{i+1}:B_{i+1}})\circ(\ident_{A_i}\otimes\vartheta_{B^t_iB_i})\big\|_\diamond\\
&\qquad\cdot\big\|\vartheta_{B^t_i}\circ\Tm_{A^t_i\to B^t_i}\big\|_\diamond\cdot\big\|(\ident_{A^t_iA_i}\otimes\vartheta_{B_i})\circ\Sm^{(i)}_{A_0\to A^t_iA_iB_i}\big\|_\diamond\\
&=\big\|\vartheta_{d_2}\circ\Tm\big\|_\diamond\cdot\big\|(\ident_{A^t_iA_i}\otimes\vartheta_{B_i})\circ\Sm^{(i)}_{A_0\to A^t_iA_iB_i}\big\|_\diamond
\end{align*}
for $i=1,\ldots,m$, and $\|(\ident_{A^t_1A_1}\otimes\vartheta_{B_1})\circ\Sm^{(1)}_{A_0\to A^t_1A_1B_1}\|_{\diamond}=\|(\ident_{A^t_1A_1}\otimes\vartheta_{B_1})\circ\Lm_{A_0\to A^t_1A_1B_1}\|_{\diamond}=1$. From these relations we obtain inductively, recalling that $A^t_{m+1}$ and $A_{m+1}$ are trivial one-dimensional systems whereas $A_0=B_{m+1}$ are $N$-dimensional and abbreviating $\Sm:=\Sm^{(m+1)}_{A_0\to B_{m+1}}:\M_N\to\M_N$:
\begin{align}
\big\|\vartheta_N\circ\Sm\big\|_\diamond\,=\,\big\|\vartheta_{B_{m+1}}\circ\Sm^{(m+1)}_{A_0\to B_{m+1}}\big\|_\diamond\,\leq\,\big\|\vartheta_{d_2}\circ\Tm\big\|_\diamond^m\,.\label{cb2waytom}
\end{align}

Next, we bound the $\diamond$-norm error $\varepsilon$ of the communication scheme (see Definition \ref{defn:TwoWayCap}) from below by evaluating at the $N$-dimensional maximally entangled state $\omega_N=\omega_{A_0R}$ between the two $N$-dimensional systems $A_0$ and $R$ and twirling over a representation of the unitary group $\U(N)$. For this we note that the twirled state is 
\begin{align*}
\int_{\U(N)}dU\,(U\otimes\overline{U})\,(\Sm\otimes\ident_N)(\omega_N)\,(U\otimes\overline{U})^\dagger=p\omega_N+(1-p)(\id_{N^2}-\omega_N)/(N^2-1)
\end{align*}
with $p:=\Trace{\omega_N\,(\Sm\otimes\ident_N)(\omega_N)}$ by Appendix \ref{Appendix}.
\begin{align*}
\varepsilon\,&=\,\frac{1}{2}\big\|\ident_N-\Sm\big\|_\diamond\,\geq\,\frac{1}{2}\big\|((\ident_N-\Sm)\otimes\ident_N)(\omega_N)\big\|_1\, \\
&=\,\frac{1}{2}\int_{\U(N)}dU\,\big\|(U\otimes\overline{U})(\omega_N-(\Sm\otimes\ident_N)(\omega_N))(U\otimes\overline{U})^\dagger\big\|_1\\
&\geq\,\frac{1}{2}\big\|\omega_N-\int_{\U(N)}dU\,(U\otimes\overline{U})\,(\Sm\otimes\ident_N)(\omega_N)\,(U\otimes\overline{U})^\dagger\big\|_1\\
&=\,\frac{1}{2}\big\|(1-p)\omega_N-(1-p)\big(\id_{N^2}-\omega_N\big)/(N^2-1)\big\|_1\,=\,1-p\,.
\end{align*}
We now derive an upper bound on $p$, by using similar steps starting from (\ref{cb2waytom}) and noting that $N(\vartheta_N\otimes\ident_N)(\omega_N)=\mathbb{F}_N$ is the flip operator:

\begin{align*}
\big\|\vartheta_{d_2}\circ\Tm\big\|_\diamond^m\,&\geq\,\big\|\vartheta_N\circ\Sm\big\|_\diamond\,\geq\,\big\|\big((\vartheta_N\circ\Sm)\otimes\ident_N\big)(\omega_N)\big\|_1\,\\
&=\,\int_{\U(N)}dU\,\big\|(\overline{U}\otimes\overline{U})\,\big((\vartheta_N\circ\Sm)\otimes\ident_N\big)(\omega_N)\,(\overline{U}^\dagger\otimes\overline{U}^\dagger)\big\|_1\\
&\geq\,\left\|\,\big(\vartheta_N\otimes\ident_N\big)\left(\int_{\U(N)}dU\,(U\otimes\overline{U})\,(\Sm\otimes\ident_N)(\omega_N)\,(U^\dagger\otimes\overline{U}^\dagger)\right)\right\|_1\\
&=\,\left\|\big(\vartheta_N\otimes\ident_N\big)\left(p\omega_N+\frac{1-p}{N^2-1}(\id_{N^2}-\omega_N)\right)\right\|_1\,\\
&=\,\left\|\frac{Np+1}{N(N+1)}\,\frac{\id_{N^2}+\mathbb{F}_N}{2}+\frac{Np-1}{N(N-1)}\,\frac{\id_{N^2}-\mathbb{F}_N}{2}\right\|_1\\
&=\,|Np+1|/2+|Np-1|/2\,\geq\,Np\,.
\end{align*}
Combining this bound with the above relation between $p$ and $\varepsilon$ yields the claim.
\end{proof}

We can now state our capacity bound:
\begin{thm}[Strong converse upper bound on the two-way capacity $\mathcal{Q}_2(\Tm)$]\label{strongconverseQ2}
Let $\Tm:\M_{d_1}\to\M_{d_2}$ be a quantum channel. Then:
\begin{align*}
\mathcal{Q}_2(\Tm)\,\leq\,\log_2\left(\|\vartheta_{d_2}\circ\Tm\|_\diamond\right)\,.
\end{align*}

Moreover, let for each $\nu\in\N$ an $(N_\nu,m_\nu,\varepsilon_\nu)$-scheme for quantum communication over $\Tm$ assisted by two-way classical communication be given in such a way that $\lim_{\nu\to\infty}m_\nu=\infty$, and define the \emph{lower code rate} $R_{\inf}:=\liminf_{\nu\to\infty}\frac{\log_2(N_\nu)}{m_\nu}$. If $R_{\inf}>\log_2\left(\|\vartheta_{d_2}\circ\Tm\|_\diamond\right)$, then the $\diamond$-norm error $\varepsilon_\nu$ of the sequence converges to $1$ (exponentially fast in $m_\nu$).
\end{thm}
\begin{proof}
To prove the first statement, suppose that a rate $R=\limsup_{\nu\to\infty}\frac{\log_2(N_\nu)}{m_\nu}>\log_2(\|\vartheta_{d_2}\circ\Tm\|_\diamond)$ is achievable by schemes with parameters $(N_\nu,m_\nu,\varepsilon_\nu)$ (cf.\ Definition \ref{defn:TwoWayCap}). Then, for any $\chi\in\R$ with $\|\vartheta_{d_2}\circ\Tm\|_\diamond<\chi<2^R$, we have $N_\nu\geq\chi^{m_\nu}$ for infinitely many values of $\nu\in\N$. Thus, by Lemma \ref{TwoWayErrorLemma},
\begin{align*}
\limsup_{\nu\to\infty}\varepsilon_\nu\,\geq\,1-\liminf_{\nu\to\infty}\frac{\|\vartheta_{d_2}\circ\Tm\|_\diamond^{m_\nu}}{N_\nu}\,\geq\,1-\liminf_{\nu\to\infty}\left(\frac{\|\vartheta_{d_2}\circ\Tm\|_\diamond}{\chi}\right)^{m_\nu}\,>\,0\,.
\end{align*}
which contradicts the requirement $\lim_{\nu\to\infty}\varepsilon_\nu=0$.

The second statement follows similarly by noting that for any $\chi<2^{R_{\inf}}$, one has $N_\nu\geq\chi^{m_\nu}$ for almost all $\nu\in\N$.
\end{proof}

The second part of Theorem \ref{strongconverseQ2} means that $\log_2(\|\vartheta_{d_2}\circ\Tm\|_\diamond)$ is not only an upper bound on the two-way capacity $\mathcal{Q}_2(\Tm)$, but even a \emph{strong converse rate} for quantum communication over $\Tm$ assisted by free two-way classical communication. This generalizes the examples in \cite[Section III]{morganwinter}, which are obtained for completely co-positive channels $\Tm$, where $\mathcal{Q}_2(\Tm)=\log_2(\|\vartheta_{d_2}\circ\Tm\|_\diamond)=0$, and for the identity channel $\Tm=\ident_{d}$, where $\log_2(\|\vartheta_d\circ\Tm\|_\diamond)=\log_2(d)=\mathcal{Q}_2(\Tm)$. The \emph{entanglement cost $E_C(\Tm)$} has been established as a strong converse rate for $\mathcal{Q}_2$ \cite{berta2012entanglement}, although it can be larger than our bound. In recent work \cite{tomamichelwildewinter} is has been shown that the upper bound $\log_2(\|\vartheta_{d_2}\circ\Tm\|_\diamond)$ from Eq.\ (\ref{equ:transBound}), and improvements thereof, are strong converse rates for the usual quantum capacity $\mathcal{Q}$ from Definition \ref{defn:QuantCap}, even when allowing for arbitrary LOCC operations at the beginning and the end of the protocol. The case of free LOCC communication during the protocol as in Definition \ref{defn:TwoWayCap} has however not been resolved in ref.\ \cite{tomamichelwildewinter}.

Even the capacity bound on $\mathcal{Q}_2(\Tm)$ from the first part of Theorem \ref{strongconverseQ2} seems to be new. In particular, an upper bound on $\mathcal{Q}_2(\Tm)$ for pure-loss bosonic channels was derived in \cite[Section 6]{6832533} based on the \emph{squashed entanglement of $\Tm$}. And while this was noted for pure-loss channels $\Tm$ to agree with the transposition bound (\ref{equ:transBound}) on $\mathcal{Q}(\Tm)$, the question was left open whether the transposition bound is a general upper bound on two-way capacity $\mathcal{Q}_2(\Tm)$.

\subsection{Strong converse rate from tensor-stable positive maps}\label{strongconverseTSsubsection}
With ideas from the proofs of Lemma \ref{TwoWayErrorLemma} and Theorem \ref{thm:CapBound}, we can use any surjective, unital and tensor-stable positive map $\Pm:\M_{d_3}\to\M_{d_2}$ that is not completely positive to derive a strong converse rate for the usual quantum capacity $\mathcal{Q}(\Tm)$ of any quantum channel $\Tm:\M_{d_1}\to\M_{d_2}$ (see Definition \ref{defn:QuantCap}). The strong converse rate we obtain is
\begin{align}\label{Q1strongconverse}
\frac{\log_2\lb\|\Pm^{-1}\circ\Tm\|_{\diamond}\,\|\Pm^{*}\lb\id_{d_2}\rb\|_{\infty}\rb\log_2(d_2)}{\log_2\lb\|(\Pm^*\otimes\ident_{d_2})(\omega_{d_2})\|_1\rb}\,,
\end{align}
which is always at least as big as our upper bound on $\mathcal{Q}(\Tm)$ from Theorem \ref{thm:CapBound}, due to $\|\Pm^{*}\|_{\diamond}\geq\|(\Pm^*\otimes\ident_{d_2})(\omega_{d_2})\|_1$. The proof that (\ref{Q1strongconverse}) is a strong converse rate for the desired task follows from the following Lemma in the same way as Theorem \ref{strongconverseQ2} follows from Lemma \ref{TwoWayErrorLemma}.
\begin{lem}Let $\Tm:\M_{d_1}\ra\M_{d_2}$ be a quantum channel and $\Pm:\M_{d_3}\ra\M_{d_2}$ be a surjective, unital and tensor-stable positive map that is not completely positive, and let $\Pm^{-1}$ be any right-inverse of $\Pm$. Let $n,m\in\N$. Then:
\begin{align*}
\inf_{\Em,\Dm}\,\frac{1}{2}\left\|\ident_{d_2}^{\otimes n} - \Dm\circ \Tm^{\otimes m}\circ \Em\right\|_\diamond\,\geq\,1-\frac{1}{d_2^{2n}}-\frac{\left(\|\Pm^*(\id_{d_2})\|_\infty\,\|\Pm^{-1}\circ\Tm\|_\diamond\right)^m+2}{\|(\Pm^*\otimes\ident_{d_2})(\omega_{d_2})\|_1^n}\,,
\end{align*}
where the infimum is over all quantum channels $\Em:\M_{d_2}^{\otimes n}\ra \M^{\otimes m}_{d_1}$ and $\Dm:\M^{\otimes m}_{d_2}\ra \M_{d_2}^{\otimes n}$.
\end{lem}
\begin{proof}Fix $\Em$, $\Dm$. As in the proof of Lemma \ref{TwoWayErrorLemma}, we bound the $\diamond$-norm with the maximally entangled state $\omega_{d_2^n}$ of dimension $d_2^n$ and then twirl:
\begin{align*}
\frac{1}{2}\left\|\ident_{d_2}^{\otimes n} - \Dm\circ \Tm^{\otimes m}\circ \Em\right\|_\diamond\,\geq\,\frac{1}{2}\left\|\omega_{d_2^n}-\int_{\U(d_2^n)}dU\,(U\otimes\overline{U})\,\rho\,(U\otimes\overline{U})^\dagger\right\|_1\,=\,1-p\,,
\end{align*}
where we denoted $\rho:=((\Dm\circ\Tm^{\otimes m}\circ\Em)\otimes\ident_{d_2^n})(\omega_{d_2^n})$, and used $\int dU\,(U\otimes\overline{U})\,\rho\,(U\otimes\overline{U})^\dagger=p\omega_{d_2^n}+(1-p)(\id_{d_2^{2n}}-\omega_{d_2^n})/(d_2^{2n}-1)$ for $p:=\Trace{\omega_N\,\rho}$. For any unitary $U\in\U(d_2^n)$ we now define the unital quantum channel $\Cm_U:\M_{d_2}^{\otimes n}\to\M_{d_2}^{\otimes n}$ by $\Cm_U(X):=UXU^\dagger$ for all $X\in\M_{d_2}^{\otimes n}$ and reuse some arguments from the proof of Theorem \ref{thm:CapBound}:\begin{align*}
\|\Pm^*&(\id_{d_2})\|_\infty^m\,\|\Pm^{-1}\circ\Tm\|_\diamond^m\,\geq\,\int_{\U(d_2^n)}dU\,\big\|(\vartheta_{d_3}\circ\Pm^*\circ\vartheta_{d_2})^{\otimes n}\circ\Cm_U\circ\Dm\circ\Pm^{\otimes m}\big\|_\diamond\,\big\|(\Pm^{-1})^{\otimes m}\circ\Tm^{\otimes m}\circ\Em\big\|_\diamond\\
&\geq\,\int_{\U(d_2^n)}dU\,\big\|\big(((\vartheta_{d_3}\circ\Pm^*\circ\vartheta_{d_2})^{\otimes n}\circ\Cm_U\circ\Dm\circ\Tm^{\otimes m}\circ\Em)\otimes\Cm_{\overline{U}}\big)(\omega_{d_2^n})\big\|_1\\
&\geq\,\left\|\left((\vartheta_{d_3}\circ\Pm^*\circ\vartheta_{d_2})^{\otimes n}\otimes\ident_{d_2^n}\right)\left(p\omega_{d_2^n}+\frac{1-p}{d_2^{2n}-1}(\id_{d_2^{2n}}-\omega_{d_2^n})\right)\right\|_1\\
&\geq\,\frac{pd_2^{2n}-1}{d_2^{2n}-1}\left\|\left((\vartheta_{d_3}\circ\Pm^*\circ\vartheta_{d_2})^{\otimes n}\otimes\ident_{d_2^n}\right)\left(\omega_{d_2^n}\right)\right\|_1-\frac{1-p}{d_2^{2n}-1}\left\|\left((\vartheta_{d_3}\circ\Pm^*\circ\vartheta_{d_2})^{\otimes n}\otimes\ident_{d_2^n}\right)\big(\id_{d_2^{2n}}\big)\right\|_1\\
&\geq\,\left(p-\frac{1}{d_2^{2n}}\right)\left\|\left(\Pm^*\otimes\ident_{d_2}\right)\left(\omega_{d_2}\right)\right\|_1^n-2\,.
\end{align*}
Converting this to an upper bound on $p$ and combining with the above relation, we obtain the claim.
\end{proof}


\section{Distillation schemes for tensor-stable positive maps}
\label{sec:Dist}

\subsection{Quantifying the distance from the completely positive maps}

For a given hermiticity-preserving map $\Pm:\M_{d_1}\ra\M_{d_2}$ we define a distance from the set of completely positive maps as
\begin{align}
d_{\text{CP}}\lb \Pm\rb := \frac{1}{2}\lb\norm{C_\Pm}{1} - \Trace{C_\Pm}\rb.
\end{align} 

By the Choi-Jamiolkowski isomorphism, $\Pm$ is completely positive iff $C_\Pm\geq0$, i.e.\ iff $d_{\text{CP}}(\Pm)=0$, whereas $d_{\text{CP}}(\Pm)>0$ otherwise. The distance $d_{\text{CP}}(\Pm)$ is just the absolute value of the sum of negative eigenvalues of the Choi matrix $C_\Pm$ of $\Pm$. The following lemma gives a useful upper bound on $d_{\text{CP}}$:

\begin{lem}

Let $\Pm:\M_{d_1}\ra\M_{d_2}$ be a positive map. If there exists a linear map $\Rm:\M_{d_1}\ra\M_{d_1}$ such that $\Rm\otimes \Pm$ is a positive map, then 
\begin{align}
d_{\text{CP}}(\Pm)\leq \norm{\ident_{d_1} - \Rm}{\diamond}\norm{\Pm}{\diamond}.
\end{align}
\label{lem:Fundamental}
\end{lem}

\begin{proof}

By elementary properties of the $\diamond$-norm and using positivity of $\Rm\otimes \Pm$ we have
\begin{align*}
\norm{\ident_{d_1} - \Rm}{\diamond}\norm{\Pm}{\diamond}&\geq\norm{\omega_{d_1}-\lb \Rm\otimes\ident_{d_1}\rb\lb\omega_{d_1}\rb}{1}\norm{\Pm}{\diamond}\geq \norm{(\ident_{d_1}\otimes \Pm)(\omega_{d_1}) - (\Rm\otimes \Pm)(\omega_{d_1})}{1}\\
&\geq\inf_{X\geq 0}\norm{C_\Pm - X}{1}.
\end{align*}
And for any hermitian matrix $H$ we have $\inf_{X\geq 0}\norm{H - X}{1} = \frac{1}{2}\lb\norm{H}{1} - \Trace{H}\rb$ by Weyl's inequalities~\cite[Corollary III.2.2]{bhatia1997matrix}.\end{proof}

Note that in the case of $\Pm:\M_{d_1}\ra\M_{d_2}$ being completely positive we can use $\Rm = \ident_{d_1}$ in order to verify that $d_{\text{CP}}(\Pm) = 0$ using Lemma \ref{lem:Fundamental}.

To apply Lemma \ref{lem:Fundamental} for an $n$-tensor-stable positive map $\Pm:\M_{d_1}\ra\M_{d_2}$ we have to find a suitable map $\Rm:\M_{d_1}\ra\M_{d_1}$ such that $\Rm\otimes \Pm$ is positive and $\Rm$ is close to the identity map. A convenient way to construct such an $\Rm$ is by considering generalized ``coding schemes'' of the form
\begin{align}
\Rm = \sum^m_{i=1} \Dm_i\circ \Pm^{\otimes(n-1)}\circ \Em_i
\label{equ:mapS}
\end{align}  
with completely positive maps $\Em_i:\M_{d_1}\ra\M_{d^{n-1}_1}$ and $\Dm_i:\M_{d^{n-1}_2}\ra\M_{d_1}$. Indeed, as $\Pm^{\otimes n}\geq 0$ we have
\begin{align*}
\Rm\otimes \Pm = \sum^m_{i=1} \lb \Dm_i\otimes \ident_{d_2}\rb \circ \lb\Pm^{\otimes n}\rb \circ \lb \Em_i\otimes \ident_{d_1}\rb\geq 0.
\end{align*}

As $\Rm\otimes \Pm$ is positive for all choices of $\Em_i$ and $\Dm_i$ in \eqref{equ:mapS} we can optimize over these completely positive maps trying to make $\norm{\ident_{d_1} - \Rm}{\diamond}$ as small as possible. This proves:

\begin{cor}

Let $\Pm:\M_{d_1}\ra\M_{d_2}$ be an $n$-tensor-stable positive map, $\Pm\neq0$. Then
\begin{align}
\frac{d_{\text{CP}}(\Pm)}{\norm{\Pm}{\diamond}}\leq \inf_{m,\Em_i,\Dm_i} \norm{\ident_{d_1} - \sum^m_{i=1} \Dm_i\circ \Pm^{\otimes(n-1)}\circ \Em_i}{\diamond},
\label{equ:CodingScheme}
\end{align}
where the infimum is taken over $m\in\N$ and completely positive maps $\Em_i,\Dm_i$.
\label{cor:Coding}
\end{cor}

The map $\Rm$ in \eqref{equ:mapS} can be interpreted as a coding scheme where quantum information is encoded by the completely positive maps $\Em_i$, sent through $(n-1)$ uses of the map $\Pm$, and decoded using the maps $\Dm_i$. The indices $i$ can be seen as classical information which is communicated from the sender to the receiver for free and without noise. A special case of this technique for $m=1$ and projectors $\Em_1,\Dm_1$ has been used in ref.\ \cite{stormer2010tensor}.

If $\Pm$ is tensor-stable positive we can take the limit $n\ra\infty$ of the approximation error on the right-hand-side of \eqref{equ:CodingScheme}. As the left-hand-side in \eqref{equ:CodingScheme} does not depend on $n$, the approximation error cannot vanish in the limit $n\ra\infty$ unless $\Pm$ is completely positive.   
 

As a first application of this idea we derive sufficient criteria for a quantum channel $\Tm:\M_{d_1}\ra\M_{d_2}$ to have $\mathcal{Q}_2\lb \Tm\rb = 0$ (see Definition \ref{defn:TwoWayCap}). For this, note that the alternating application of the LOCC maps and the $m$ channel uses in Eq.\ (\ref{LOCCcodingscheme}) can be written as
\begin{align*}
\Lm_{A_m:B^t_mB_m\to B_{m+1}}\circ\ldots\circ\Tm_{A^t_1\to B^t_1}\circ\Lm_{A_0\to A^t_1A_1:B_1}\left(\rho\right)~=~\sum_k K^B_k\,\, T^{\otimes m}\left(K^A_k\,\rho\,(K^A_k)^\dagger\right)\,(K^B_k)^\dagger\,,
\end{align*}
where here the $K^A_k$ (with a multi-index $k$) are simply all the time-ordered products of Kraus operators $(K^A_i\otimes\ket{j})$ on the sender's side from (\ref{LOCCkraus}) occurring in the LOCC maps in (\ref{LOCCcodingscheme}); similarly for $K^B_k$ on the receiver's side. Thus, by defining completely positive maps $\Em_k:\M_{A_0}\to\M_{d_1}^{\otimes m}$ and $\Dm_k:\M_{d_2}^{\otimes m}\to\M_{B_{m+1}}$ by $\Em_k(X):=K_k^A\,X\,(K_k^A)^\dagger$ and $\Dm_k(Y):=K_k^B\,Y\,(K_k^B)^\dagger$, we have shown the existence of completely positive maps $\Em_k$, $\Dm_k$ such that
\begin{align*}
\|\ident_{N}-\sum_{k}\Dm_k\circ\Tm^{\otimes m}\circ\Em_k\|_\diamond\,=\,2\varepsilon\,,
\end{align*}
whenever there exists a $(N,m,\varepsilon)$-scheme for LOCC-assisted quantum communication according to Definition \ref{defn:TwoWayCap}. If the quantum channel $\Tm$ has positive two-way capacity $\mathcal{Q}_2(\Tm)>0$, then for any fixed $N$, one can certainly transmit an $N$-dimensional quantum system with arbitrarily low error ($\varepsilon\to0$) in the limit of arbitrarily many channel uses ($m\to\infty$).

Now consider the case where the quantum channel $\Tm$ with $\mathcal{Q}_2(\Tm)>0$ is of the form $\Tm = \sum_j\Vm_j\circ\Pm\circ\Wm_j$ for a tensor-stable positive map $\Pm:\M_{d_3}\ra\M_{d_4}$ that is \emph{not} completely positive and for completely positive maps $\Vm_j:\M_{d_1}\ra\M_{d_3}$ and $\Wm_j:\M_{d_4}\to\M_{d_2}$. Then, by setting $N:=d_3$ in the previous paragraph, we have
\begin{align*}
0~&=\,\lim_{m\ra\infty}\inf_{\Dm_k,\Em_k}\|\ident_{d_3}-\sum_{k}\Dm_k\circ\Tm^{\otimes m}\circ\Em_k\|_\diamond\\
&=\,\lim_{m\ra\infty}\inf_{\Dm_k,\Em_k}\Big\|\ident_{d_3}- \sum_{k,j_1,\ldots,j_m}\left(\Dm_k\circ\bigotimes_{\ell=1}^m\Vm_{j_\ell}\right)\circ \Pm^{\otimes m}\circ\left(\bigotimes_{\ell=1}^m\Wm_{j_\ell}\circ \Em_k\right)\Big\|_{\diamond}\,. 
\end{align*}
But this leads to a contradiction since, by interpreting $\widetilde{\Dm}_i:=\Dm_k\circ(\bigotimes_{\ell=1}^m\Vm_{j_\ell})$ and $\widetilde{\Em}_i:=(\bigotimes_{\ell=1}^m\Wm_{j_\ell})\circ \Em_k$ with the multi-index $i\equiv(k,j_1,\ldots,j_m)$ as the encoding and decoding maps for the map $\Pm$, Corollary \ref{cor:Coding} would imply that $d_{\rm CP}(\Pm)=0$, meaning that $\Pm$ would be completely positive contrary to assumption. This proves the following:
\begin{cor}Let $\Tm$ be a quantum channel of the form $\Tm = \sum_j\Vm_j\circ\Pm\circ\Wm_j$ for a tensor-stable positive map $\Pm$ that is not completely positive and for completely positive maps $\Vm_j$, $\Wm_j$. Then:
\begin{align*}
\mathcal{Q}_2\lb\Tm\rb = 0.
\end{align*}
\label{cor:Q20}
\end{cor}
The special case $\Tm=\Vm\circ\vartheta$ of this theorem, i.e.\ where $\Tm$ is completely co-positive, was already established in \cite{rains2001semidefinite}. It however appears that Corollary \ref{cor:Q20} could give new channels $\Tm=\Pm\circ\Wm$ or $\Tm=\Vm\circ\Pm$ with $\mathcal{Q}(\Tm)=0$ beyond Theorem \ref{thm:CapBound} and Eq.\ (\ref{otherversionQbound}), at least when $\Pm$ does not possess a right- or left-inverse.


Using Corollary \ref{cor:Q20} one can  show that any \emph{non-trivial} tensor-stable positive map $\Pm:\M_{d_1}\to\M_{d_2}$ will immediately yield new channels $\Tm:\M_d\to\M_d$ with $\mathcal{Q}_2(\Tm)=0$ (for both $d=d_1$ and $d=d_2$). To see this, note that by writing the separable map $\Sm$ from Lemma \ref{thm:ConnectionWerner} into single Kraus operators as in Section \ref{proofofoneparameterfamily}, we can construct completely positive maps $\Vm_j$, $\Wm_j$ such that $\sum_j\Vm_j\circ\Pm\circ\Wm_j=\Wm_{\widetilde{p}}$, where $\Wm_{\widetilde{p}}:\M_d\to\M_d$ with $\widetilde{p}\in[-1,0)$ is a quantum channel from the family (\ref{wernerchannel}) whose Choi matrix is an entangled Werner state. Thus, by Corollary \ref{cor:Q20} and the depolarizing idea from Section \ref{proofofoneparameterfamily}, all the channels $\Wm_p$ with $p\in[\widetilde{p},0)$ have vanishing two-way capacity $\mathcal{Q}_2(\Wm_p)=0$, although these channels are not detected as such by the existing criterion from \cite{rains2001semidefinite} (or by Theorem \ref{strongconverseQ2}) as they are not completely co-positive. The channels constructed in this way are however already known to have vanishing one-way quantum capacity $\mathcal{Q}(\Wm_p)=0$, since they possess a symmetric extension (note, the case $d\leq2$ does not occur here due to Theorem \ref{thm:NoQubitMap}) and are thus anti-degradable \cite{bruss1998optimal,bennett1997capacities,PhysRevLett.108.230507}.


In the following chapters we will use another way of thinking about coding schemes of the form \eqref{equ:mapS}. Recall that a completely positive map $\Sm:\M_{d_1}\otimes\M_{d_3}\ra \M_{d_2}\otimes\M_{d_4}$ is called separable if its Kraus operators are product operators $\lset A_i\otimes B_i\rset^m_{i=1}$, i.e. $\Sm\lb X\rb = \sum^m_{i=1} \lb A_i\otimes B_i\rb X \lb A_i\otimes B_i\rb^\dagger$ for all $X\in\M_{d_1}\otimes \M_{d_3}$. 

The application of a separable map $\Sm:\M_{d_1^{n-1}}\otimes\M_{d_2^{n-1}}\ra \M_{d_1}\otimes\M_{d_1}$ to $(n-1)$ copies of the Choi-matrix $C_\Pm$ of some linear map $\Pm:\M_{d_1}\ra\M_{d_2}$ corresponds via the Choi-Jamiolkowski isomorphism to a map $\Rm$ as
\begin{align*}
C_\Rm = \Sm\lb C_\Pm^{\otimes (n-1)}\rb,
\end{align*}
with $\Rm\lb X\rb = \sum^m_{i=1} B_i \Pm^{\otimes (n-1)}\lb A^T_i X \overline{A_i}\rb B^\dagger_i$. The map $\Rm$ is of the form \eqref{equ:mapS}, which by slightly modifying the proof of Lemma \ref{lem:Fundamental} implies:

\begin{cor}\label{cor:DistChoiMat}

Let $\Pm:\M_{d_1}\ra\M_{d_2}$ be an $n$-tensor-stable positive map. Then
\begin{align}
\frac{d_{\text{CP}}(\Pm)}{\norm{\Pm}{\diamond}}\leq \inf_{\Sm\,\text{sep}} \norm{\omega_{d_1} - \Sm\lb C_\Pm^{\otimes (n-1)}\rb}{1},
\label{equ:DistError}
\end{align}
where the infimum is taken over all separable completely positive maps $\Sm:\M_{d_1^{n-1}}\otimes\M_{d_2^{n-1}}\ra \M_{d_1}\otimes\M_{d_1}$.
\label{cor:Fund}
\end{cor} 

If the Choi-matrix $C_\Pm$ is a quantum state, then the problem of finding separable maps $\Sm$ to minimize the error on the right-hand-side of \eqref{equ:DistError} is well-studied in quantum information theory: A state $C_\Pm\in\M_{d_1}\otimes\M_{d_2}$ is distillable iff there exists a sequence of LOCC-maps $\Sm_n$ such that $\Sm_n\lb C^{\otimes n}_\Pm\rb\ra\omega_{d_1}$. As LOCC-maps are in particular separable this sequence leads to a vanishing (in the limit $n\ra\infty$) right-hand-side in \eqref{equ:DistError}. 

Note that any positive map that is not completely co-positive has an NPPT, but not necessarily positive, Choi-matrix. We generalize distillation schemes from quantum states to arbitrary block-positive matrices to show (using Corollary \ref{cor:DistChoiMat}) that tensor-stable positivity implies complete positivity for certain classes of non-completely co-positive maps.

\subsection{Proof of Theorem \ref{thm:NoQubitMap} and Theorem \ref{thm:NPTImpl}}
\label{SectionB}

To prove Theorem \ref{thm:NoQubitMap} and Theorem \ref{thm:NPTImpl} we will use the theory of entanglement distillation. For convenience we collect some basic definitions and results in Appendix \ref{Appendix}. The central result we will need is Lemma \ref{lem:TwirlBlock}, which shows that applying the twirl~\citep{werner1989quantum} to a block-positive and NPPT matrix yields (up to normalization) a Werner state, i.e. it yields in particular a positive matrix. This allows us to extend the theory of entanglement distillation to block-positive matrices. We will start with a basic lemma:  

\begin{lem}[Werner states from positive maps]

Let $\Pm:\M_{d_1}\ra \M_{d_2}$ be a positive map and $d\in\lset d_1 ,d_2\rset$. If $\Pm$ is not completely co-positive, then there exists a separable completely positive map $\Sm:\M_{d_1}\otimes\M_{d_2}\ra \M_{d}\otimes\M_d$ such that $\Sm\lb C_\Pm\rb$ is an entangled $d$-dimensional Werner state (see Appendix \ref{Appendix}).
\label{thm:ConnectionWerner}
\end{lem}

\begin{proof} 

This proof works similar to the protocol introduced in \cite{PhysRevA.59.4206} for states. Consider $d=d_2$ now, and we will treat the case $d=d_1$ later. As $\Pm$ is not completely co-positive, there exists a normalized vector $\ket{\psi}\in\C^{d_1}\otimes \C^{d_2}$ with $\bra{\psi}C^{T_2}_\Pm\ket{\psi}<0$. Express this vector as $\ket{\psi} = (A\otimes\id_{d_2})\ket{\Omega_{d_2}}$ for some $d_1\times d_2$ matrix $A$ and the maximally entangled state $\ket{\Omega_{d_2}}$. Now define a new linear map $\Pm'$ via the Choi-Jamiolkowski isomorphism by applying a local filtering operation 
\begin{align*}
C_{\Pm'} := (A^\dagger\otimes\id_{d_2})C_\Pm(A\otimes\id_{d_2})\in\M_{d_2}\otimes\M_{d_2},
\end{align*}
i.e. $\Pm'\lb X\rb := \Pm\lb \overline{A} X A^T\rb$ for all $X\in\M_{d_2}$.
 
The matrix $C_{\Pm'}$ is block-positive and fulfills $\Trace{C_{\Pm'}\mathbb{F}_{d_2}} = d_2\bra{\psi}C^{T_2}_\Pm\ket{\psi}<0$. Therefore we can use Lemma \ref{lem:TwirlBlock} and conclude that applying the $UU$-twirl leads to a positive matrix. After normalization we obtain a Werner state
\begin{align*}
\rho_W = \frac{1}{\Trace{C_{\Pm'}}}\int_{U\in\Um(d_2)} \lb U\otimes U\rb C_{\Pm'} \lb U\otimes U\rb^\dagger \text{dU}\in\M_{d_2}\otimes\M_{d_2}.
\end{align*}
Due to $\Trace{C_{\Pm'}\mathbb{F}_{d_2}}<0$, this state is entangled. Finally, the composition of the twirl (which is separable, see Appendix \ref{Appendix}) with the filtering map is a separable completely positive map. 

If one chooses $d=d_1$, then write $\ket{\psi}=(\id_{d_1}\otimes B)\ket{\Omega_{d_1}}$ with a $d_2\times d_1$-matrix $B$, and define $C_{\Pm'}:=(\id_{d_1}\otimes B^T)C_\Pm(\id_{d_1}\otimes\overline{B})$. The proof goes then through similarly.

\end{proof}

From this Lemma we get:

\begin{proof}[proof of Theorem \ref{thm:NPTImpl}]

For $d\in\lset d_1,d_2\rset$ the Choi-matrix $C_\Pm$ of every non-trivial tensor-stable positive map $\Pm:\M_{d_1}\ra\M_{d_2}$ yields an entangled Werner state by the application of a separable completely positive map $\Sm$ according to Lemma \ref{thm:ConnectionWerner}. If this Werner state is distillable, there exists a sequence of separable (even LOCC) completely positive maps $(\Sm_n)_{n\in\N}$ such that $\norm{\omega_{d} - \Sm_n\circ \Sm^{\otimes n}\lb C_\Pm^{\otimes n}\rb}{1}\ra 0$ as $n\ra\infty$. But then Corollary \ref{cor:Fund} implies that $\Pm$ is completely positive contradicting the assumptions. 

\end{proof}

\begin{proof}[proof of Theorem \ref{thm:NoQubitMap}]

As all entangled Werner states on $\M_2\otimes \M_2$ are distillable~\cite{PhysRevLett.80.5239,PhysRevA.59.4206} there is no non-trivial tensor-stable positive map $\Pm:\M_2\ra\M_d$ or $\Pm:\M_d\ra\M_2$ for $d\in\N$ according to Theorem \ref{thm:NPTImpl}.  

\end{proof}

\subsection{Proof of Theorem \ref{thm:OneParamFambl}}\label{proofofoneparameterfamily}

Using the techniques from section \ref{SectionB} we can define one-parameter families of non-trivial positive maps such that there exists a non-trivial tensor-stable positive map iff it exists within this family.  

\begin{proof}[proof of Theorem \ref{thm:OneParamFambl}]

\emph{ad (ii):} For $p\in\lbr -1,0\rb$ the Werner state $\rho^{(p)}_W\in\M_{d^2}$ is NPPT. Therefore the map $\Pm_p :\M_{d^2}\ra\M_{d^2}$ is neither completely positive nor completely co-positive, as its Choi-matrix is $\rho^{(p)}_W\otimes \lb\rho^{(p)}_W\rb^{T_2}$. 

\emph{ad (i):} For a non-trivial tensor-stable positive map $\Pm:\M_{d_1}\ra\M_{d_2}$, neither $\Pm$ nor $\vartheta_{d_2}\circ\Pm$ are completely co-positive. According to Lemma \ref{thm:ConnectionWerner} there exist $p_1,p_2\in \lbr-1,0\rb$ and separable completely positive maps $\Sm_1,\Sm_2:\M_{d_1}\otimes\M_{d_2}\to\M_d\otimes\M_d$ such that 
\begin{align*}
\rho^{(p_1)}_W&=\Sm_1(C_\Pm)\,, \\
\rho^{(p_2)}_W&=\Sm_2(C_{\vartheta_{d_2}\circ\Pm}) = \Sm_2\circ(\ident_{d_1}\otimes\vartheta_{d_2})(C_\Pm).
\end{align*} 
It is obvious that for the separable completely positive map $\Sm_2(X)=\sum_i(A_i\otimes B_i)X(A_i\otimes B_i)^\dagger$ the map $\tilde{\Sm}_2 = (\ident_d\otimes \vartheta_d) \circ \Sm_2\circ (\ident_{d_1}\otimes\vartheta_{d_2})$ is again separable completely positive. Thus, the separable map $\Sm_1\otimes\tilde{\Sm}_2$ applied to two tensor copies of $C_\Pm$ gives:
\begin{align*}
\big(\Sm_1\otimes\tilde{\Sm}_2\big)(C_\Pm\otimes C_\Pm)=\rho^{(p_1)}_W\otimes \lb\rho^{(p_2)}_W\rb^{T_2}\,.
\end{align*}

By applying a depolarizing channel $\Dm_\alpha:\M_{d}\ra\M_{d}$ of the form $\Dm_\alpha(X) = (1-\alpha)\Trace{X}\frac{\id_{d}}{d} + \alpha X$ (with $\alpha$ chosen appropriately) to one half of either $\rho^{(p_1)}_W$ (if $p_1 < p_2$) or $\lb\rho^{(p_2)}_W\rb^{T_2}$ (if $p_1 > p_2$) we can increase the corresponding parameter to obtain the desired state $\rho^{(p)}_W\otimes \lb\rho^{(p)}_W\rb^{T_2}$ with $p = \max\lb p_1 ,p_2\rb < 0$. Thus, there exists a separable and completely positive map $\Rm:(\M_{d_1}\otimes\M_{d_2})^{\otimes 2}\to(\M_d\otimes\M_d)^{\otimes 2}$, given by the composition of $\Sm_1\otimes\tilde{\Sm}_2$ with $(\Dm_\alpha\otimes\ident_{d})\otimes\ident_{d^2}$ or $\ident_{d^2}\otimes(\Dm_\alpha\otimes\ident_{d})$, such that $C_{\Pm_p}=\rho^{(p)}_W\otimes \lb\rho^{(p)}_W\rb^{T_2}=\Rm(C_{\Pm^{\otimes 2}})=\sum_i(C_i\otimes D_i)C_{\Pm^{\otimes2}}(C_i\otimes D_i)^\dagger$ where $\Pm_p$ was defined in \eqref{equ:OneParamFamMap}. By the Choi-Jamiolkowski isomorphism we can thus write $\Pm_{p}(X) = \sum_{i}D_i \Pm^{\otimes 2}(C_i^T X \overline{C}_i) D_i^\dagger$, which shows that $\Pm_p$ is tensor-stable positive as $\Pm$ was.
\end{proof}

Note that the construction from the proof of Theorem \ref{thm:OneParamFambl} also works for an $n$-tensor-stable positive map $\Pm:\M_{d_1}\ra\M_{d_2}$. The positive map $\Pm_p$ of the from \eqref{equ:OneParamFamMap} obtained this way is $\lfloor\frac{n}{2}\rfloor$-tensor-stable positive.

\subsection{Generalization of the reduction criterion}

In this section we will generalize the reduction criterion and use the well-known recurrence protocol~\cite{PhysRevA.59.4206} to prove bounds on $d_{\text{CP}}\lb\Pm\rb$ for an $n$-tensor-stable positive map $\Pm$. 

We will need the following lemma (an analogue of Lemma \ref{thm:ConnectionWerner}):

\begin{lem}[Reduction criterion]

Let $\Pm:\M_{d_1}\ra\M_{d_2}$ be a positive map. Let $\Gamma_d:\M_d\ra\M_d$ denote the reduction map $\Gamma(X):= \Trace{X}\id_d - X$. Then we have:
\begin{enumerate} 
\item If $\Gamma_{d_2}\circ \Pm$ is not completely positive there exists a separable completely positive map $\Sm:\M_{d_1}\otimes\M_{d_2}\ra \M_{d_2}\otimes\M_{d_2}$ s.th. $\Sm\lb C_\Pm\rb$ is an entangled isotropic state (see Appendix \ref{Appendix}). 
\item If $\Pm\circ \Gamma_{d_1}$ is not completely positive there exists a separable completely positive map $\Sm:\M_{d_1}\otimes\M_{d_2}\ra \M_{d_1}\otimes\M_{d_1}$ s.th. $\Sm\lb C_\Pm\rb$ is an entangled isotropic state (see Appendix \ref{Appendix}). 
\end{enumerate}

\label{thm:Red}
\end{lem}

\begin{proof}

Again the proof works similar to the protocol introduced in \cite{PhysRevA.59.4206} for states. We will start with the first case.

As $\Gamma_{d_2}\circ \Pm$ is not completely positive the Choi-matrix $C_{\Gamma_{d_2}\circ\Pm} = \frac{1}{d_1}\Pm^*\lb\id_{d_2}\rb^T\otimes \id_{d_2} - C_\Pm$, derived using Lemma \ref{Lemma:tricks}, is not positive. Thus, there exists a normalized vector $\ket{\psi}\in\C^{d_1}\otimes\C^{d_2}$ with 
\begin{align*}
\frac{1}{d_
1}\bra{\psi}\Pm^*(\id_{d_2})^T\otimes \id_{d_2}\ket{\psi} < \bra{\psi}C_\Pm\ket{\psi}.
\end{align*}
Express this vector as $\ket{\psi} = (A\otimes \id_{d_2})\ket{\Omega_{d_2}}$ for some $d_1\times d_2$-matrix $A$ and define a new linear map $\Pm'$ with Choi matrix 
\begin{align*}
C_{\Pm'} = (A^\dagger\otimes \id_{d_2})C_\Pm(A\otimes \id_{d_2})\in\M_{d_2}\otimes \M_{d_2}.
\end{align*}
Note that by construction $\Pm'$ is a positive map obtained from $\Pm$ via a separable (even local) completely positive map. Furthermore we have using Lemma \ref{Lemma:tricks}
\begin{align*}
\Trace{C_{\Pm'}} =& \bra{\Omega_{d_1}} AA^\dagger\otimes \Pm^*(\id_{d_2})\ket{\Omega_{d_1}} \\
=& \bra{\Omega_{d_2}} A^\dagger \Pm^*(\id_{d_2})^T A\otimes \id_{d_2}\ket{\Omega_{d_2}}\cdot\frac{d_2}{d_1} \\
=& \bra{\psi}\Pm^*(\id_{d_2})^T\otimes \id_{d_2}\ket{\psi}\cdot\frac{d_2}{d_1} < d_2\bra{\psi}C_{\Pm}\ket{\psi}.
\end{align*}
Therefore we have $\Trace{C_{\Pm'}\omega_{d_2}} = \bra{\psi}C_\Pm\ket{\psi}>\frac{\Trace{C_{\Pm'}}}{d_2} > 0$. Note that $\Trace{C_{\Pm'}}>0$ as $C_{\Pm'}$ is block-positive and $C_{\Pm'}\neq 0$ as $\Gamma_{d_2}\circ\Pm$ is not completely positive. By applying Lemma \ref{lem:TwirlBlock} we conclude that
\begin{align*}
\rho^{(p)}_I = \frac{1}{\Trace{C_{\Pm'}}}\int_{U\in\Um(d_2)} \lb U\otimes \overline{U}\rb C_{\Pm'} \lb U\otimes \overline{U}\rb^\dagger \text{dU}\in\M_{d_2}\otimes\M_{d_2}.
\end{align*} 
is an isotropic state, with $p = \frac{\Trace{C_{\Pm'}\omega_{d_2}}}{\Trace{C_{\Pm'}}} > \frac{1}{d_2}$. Thus this state is entangled and the composition of the twirl (which is a separable completely positive map, see Appendix \ref{Appendix}) with the filtering map is separable and completely positive.

The second part works similar to the first part. Note that $\Pm\circ \Gamma_{d_1}$ not being completely positive is equivalent to the existence of a normalized vector $\ket{\psi}\in\C^{d_1}\otimes\C^{d_2}$ with 
\begin{align*}
\frac{1}{d_
1}\bra{\psi}\id_{d_1}\otimes \Pm\lb\id_{d_1}\rb\ket{\psi} < \bra{\psi}C_\Pm\ket{\psi}.
\end{align*}   
Express this vector as $\ket{\psi} = (\id_{d_1}\otimes B)\ket{\Omega_{d_1}}$ for some $d_2\times d_1$-matrix $B$ and define a new linear map $\Pm'$ with Choi matrix 
\begin{align*}
C_{\Pm'} = (\id_{d_1}\otimes B^\dagger)C_\Pm(\id_{d_1}\otimes B)\in\M_{d_1}\otimes \M_{d_1}.
\end{align*}
Now by a similar calculation as before we have $\Trace{C_{\Pm'}} = \bra{\psi}\id_{d_1}\otimes \Pm\lb\id_{d_1}\rb\ket{\psi} < d_1\bra{\psi}C_\Pm\ket{\psi}$. The rest of the proof works the same as for the first case.  
 
\end{proof}

Lemma \ref{thm:Red} shows how to obtain an entangled isotropic state from the Choi-matrix of a positive map violating the reduction criterion. It is well-known that these states are distillable by the recurrence protocol~\cite{PhysRevA.59.4206}. More precisely there exists a separable completely positive map $\Sm:\M_{d^2}\ra\M_d$ with 
\begin{align}
T_{U\overline{U}}\circ \Sm\lb(\rho^{(p)}_{I})^{\otimes 2}\rb = \rho^{(r(p))}_{I},
\label{equ:Recurrence}
\end{align}  
where $T_{U\overline{U}}$ denotes the $U\overline{U}$-twirl and where 
\begin{align*}
r(p) = \frac{1+p\lb pd(d^2 + d -1) - 2\rb}{p^2d^3 - 2pd + d^2 + d -1}.
\end{align*}
It can be easily seen that for $p>\frac{1}{d}$ we have $r^{(m)}(p)\ra 1$ as $m\ra\infty$, where the notation $r^{(m)}(p)$ means that we concatenate $m$ applications of the function $r$, i.e.\ $r^{(m)}(p):=r(r(\ldots r(p)))$. Therefore iterating the protocol using up many copies of the input state $\rho^{(p)}_{I}$ leads to isotropic states close to the maximally entangled state $\omega_d$. In the following we use this protocol and Corollary \ref{cor:Fund} to upper-bound the distance of an $n$-tensor-stable positive map violating the reduction criterion to the cone of completely positive maps.   

Note that the original protocol~\cite{PhysRevA.59.4206} has a sufficiently small but non-zero probability of failure. As the separable completely positive maps $\Sm$ in Corollary \ref{cor:Fund} do not have to be trace-preserving we can avoid the possibility of failure by choosing only the Kraus operators corresponding to a successful measurement for $\Sm$.

\begin{thm}[Bound from the recurrence protocol]

Let $\Pm:\M_{d_1}\ra\M_{d_2}$ be a positive map and such that $\Pm\circ \Gamma_{d_1}$ is not completely positive, i.e.   
\begin{align}
p\,&:=\,\sup_{\ket{\psi}\in\C^{d_1}\otimes\C^{d_2}}\frac{\bra{\psi}C_\Pm\ket{\psi}}{\bra{\psi}\id_{d_1}\otimes \Pm\lb\id_{d_1}\rb\ket{\psi}}\,\\
&=\,\lambda_{max}\big[\big(\id_{d_1}\otimes \Pm\lb\id_{d_1}\rb\big)^{-1/2}C_\Pm\big(\id_{d_1}\otimes \Pm\lb\id_{d_1}\rb\big)^{-1/2}\big]~\in~(1/{d_1},1]\,,
\label{equ:Red}
\end{align}
using generalized inverses and denoting by $\lambda_{max}[\,\cdot\,]$  the maximum eigenvalue.
If $\Pm$ is $n$-tensor-stable positive, then
\begin{align}
d_{\text{CP}}(\Pm) \leq 2(1-p)\left(g_{d_1}(p)\right)^{\lfloor\log_2(n-1)\rfloor}
\end{align}
where $g_d(p):= \frac{d(d+1) - 2 - p(pd(d-1) + 2(d-1))}{(1-p)(p(pd^3 - 2d) + d^2 + d -1)}$. Note that $g_d(p)\in[0,1)$ for $p\in(\frac{1}{d},1]$.
\label{thm:BoundRecurrence}
\end{thm}

\begin{proof}

By Lemma \ref{thm:Red} (and its proof) there is a separable completely positive map $\Sm_1:\M_{(d_1d_2)^{(n-1)}}\ra\M_{(d^1_2)^{(n-1)}}$ with $\Sm_1(C^{\otimes (n-1)}_\Pm) = (\rho^{(p)}_I)^{\otimes (n-1)}$. By \eqref{equ:Recurrence} we can apply the recurrence protocol for $\lfloor\log_2(n-1)\rfloor$ levels yielding $\rho^{(p')}_I\in\M_{d^2_1}$ with $p' = r^{(\lfloor\log_2(n-1)\rfloor)}(p)$. Composing these two protocols gives a separable completely positive map $\Sm:\M_{(d_1d_2)^{(n-1)}}\ra\M_{(d^2_1)}$ with 
\begin{align*}
\norm{\omega_{d_1} - \Sm\lb C^{\otimes (n-1)}_\Pm\rb}{1}= 2(1-p').
\end{align*}
A simple calculation gives 
\begin{align*}
1-p' = 1-r^{(\lfloor\log_2(n-1)\rfloor)}(p)\leq \lb g_{d_1}(p)\rb^{\lfloor\log_2(n-1)\rfloor}(1-p),
\end{align*}
since $g_{d_1}(p) = \frac{1-r(p)}{1-p}$ is strictly monotonously decreasing for $p\in(\frac{1}{d_1},1)$ and is equal to the expression above. 

By Corollary \ref{cor:Fund} we finally have
\begin{align*}
\frac{d_{\text{CP}}(T)}{\norm{\Pm}{\diamond}}\leq 2(1-p)\lb g_{d_1}(p)\rb^{\lfloor\log_2(n-1)\rfloor}.
\end{align*}

\end{proof}

For dimension $d_1 = 2$ we have $\Gamma_2 = \Um\circ \vartheta_2$ for some unitary conjugation $\Um:\M_2\ra\M_2$. Therefore the positive maps $\Pm:\M_d\ra\M_2$ such that $\Gamma_2\circ\Pm$ is completely positive are precisely the completely co-positive maps. For general maps $\Pm:\M_{d_1}\ra\M_{d_2}$, if $\Gamma_{d_2}\circ\Pm$ is not completely positive, then $\vartheta_{d_2}\circ\Pm$ is not completely positive, i.e. $\Pm$ is not completely co-positive.

\section{Conclusion}
We have introduced the notions of \emph{$n$-tensor-stable positive} and \emph{tensor-stable positive maps}, and have investigated whether such maps exist outside of the cones of completely positive or completely co-positive maps. We showed that tensor-stable positive maps outside these families would provide novel bounds on the quantum capacity of quantum channels. Our main technique was to apply coding schemes from distillation theory to block-positive operators rather than to density matrices. Thereby and by the Choi correspondence between block-positive operators and positive maps, we related the existence of tensor-stable positive maps to the existence of NPPT bound entanglement. We also showed that the cb-norm bound coming from the transposition map yields a strong converse rate for the two-way quantum capacity $\mathcal{Q}_2$, and established strong converse rates on the usual quantum capacity $\mathcal{Q}$ coming from other tensor-stable positive maps.

The main question left open by our work is whether non-trivial tensor-stable positive maps exist at all, i.e.\ maps outside of the above cones that are $n$-tensor-stable positive for all $n\in\N$. We have reduced this existence question to certain one-parameter families of candidate maps (Theorem \ref{thm:OneParamFambl}). But can this reduction be used to decide the existence, or at least to prove the non-existence result of Theorem \ref{thm:NoQubitMap} directly?

Furthermore, the converse of Theorem \ref{thm:NPTImpl} is open: Does the existence of NPPT bound entanglement imply the existence of non-trivial tensor-stable positive maps? Note that such an equivalence would be rather different from the equivalence result of \cite{PhysRevA.61.062312}, linking NPPT bound entanglement to completely co-positive maps which are not completely positive but such that all their tensor powers are 2-positive. The map of interest in the latter scenario lies within the completely co-positive cone, which is among the trivial cases for our work.

Our existence result of an $n$-tensor stable positive map for every $n \in \N$ (Theorem \ref{ntensor}) is analogous to  the result in the theory of entanglement distillation which guarantees for every $n$ the existence of NPPT states that are not $n$-copy distillable \cite{PhysRevA.61.062312,PhysRevA.61.062313}. Our Lemma \ref{cubittlemma} however appears to be too weak to show the existence of a map that is $n$-tensor stable positive for all $n$, see Eq.\ \eqref{equ:ntsBound}.


Finally, note that one may relate the existence of a tensor-stable positive map to the stability of operator norms under tensor products \cite{paulsen2002completely}: a positive unital map $\Tm:\M_{d_1}\to\M_{d_2}$ is $n$-tensor stable positive if and only if the induced operator norm $\|\Tm^{\otimes n}\|_{\infty\to\infty}=1$.

\bigskip

{\bf Acknowledgements.} We thank Sergey Filippov and Milan Mosonyi for useful discussions and the Isaac Newton Institute for the Mathematical Sciences  for its hospitality and support during the programme ``Quantum Control Engineering'' (M.M.W.). We acknowledge financial support from the CHIST-ERA/BMBF project CQC (A.M.H.\ and M.M.W.), the Marie Curie Intra-European Fellowship QUINTYL (D.R.), and the John Templeton Foundation (D.R.\ and M.M.W.). The opinions expressed in this publication are those of the authors and do not necessarily reflect the views of the John Templeton
Foundation.

\appendix
\section{Twirling and families of symmetric matrices}
\label{Appendix}

The main ingredient in the distillation protocols we will apply is the $UU$\textbf{-twirl} operation $T_{UU}:\M_{d}\otimes \M_{d}\ra\M_{d}\otimes \M_{d}$~\cite{werner1989quantum}, defined as
\begin{align*}
T_{UU}\lb X\rb := \int_{U\in\Um(d)} \lb U\otimes U\rb X \lb U\otimes U\rb^\dagger \text{d}U.
\end{align*} 
An application of the Schur-Weyl duality gives~\cite{werner1989quantum}  (for $d\geq2$)
\begin{align}
\int_{U\in\Um(d)} \lb U\otimes U\rb X \lb U\otimes U\rb^\dagger \text{dU} = \lbr \frac{\Trace{X}}{d^2 - 1} - \frac{\Trace{X\mathbb{F}_d}}{d(d^2 - 1)}\rbr(\id_d\otimes\id_d) - \lbr \frac{\Trace{X}}{d(d^2 - 1)} - \frac{\Trace{X\mathbb{F}_d}}{d^2 - 1}\rbr\mathbb{F}_d.
\label{equ:twirl}
\end{align}
It is easy to verify that this matrix is positive iff $\Trace{X\mathbb{F}_d}\in\lbr-\Trace{X},\Trace{X}\rbr$ (and $\Trace{X}\geq0$); the twirled matrix has positive partial transpose iff $\Trace{X\mathbb{F}_d}\in\lbr 0, d\Trace{X}\rbr$ (and $\Trace{X}\geq0$).

Using unitary 2-designs~\cite{roy2009unitary} it is well-known that the twirl is a separable completely positive map, i.e. there exists a finite set of unitary product matrices $\lset U_i\otimes U_i\rset^m_{i=1}$ such that $T_{UU}\lb X\rb = \frac{1}{m}\sum^m_{i=1} \lb U_i\otimes U_i\rb X\lb U_i\otimes U_i\rb^\dagger$. 

States of the form \eqref{equ:twirl} are clearly invariant under the $UU$-twirl operation and are called \textbf{Werner states}~\cite{werner1989quantum}. We denote these states by $\rho^{(p)}_W$, parametrized by $p:= \Trace{\rho^{(p)}_W\mathbb{F}_d}\in\lbr -1,1\rbr$ and satisfying $\Trace{\rho^{(p)}_W}=1$. It is well-known that these states are entangled (and NPPT) for $p\in[-1,0)$ and separable for $p\in[0,1]$ (thus, PPT). Furthermore for $d=2$ all entangled Werner states are distillable~\cite{PhysRevLett.80.5239,PhysRevA.59.4206}. But for $d>2$ it is \emph{not} known whether all entangled Werner states are distillable. 

By partially transposing the matrices of form $\eqref{equ:twirl}$ we obtain matrices invariant under the $U\overline{U}$-twirl operation, i.e. invariant under the operation
\begin{align*}
X\mapsto\int_{U\in\Um(d)} \lb U\otimes \overline{U}\rb X \lb U\otimes \overline{U}\rb^\dagger \text{d}U.
\end{align*}
We will denote the states obtained in this way by $\rho^{(p)}_I$, which are the \textbf{isotropic states}~\cite{PhysRevA.59.4206} parametrized by $p:=\Trace{\rho^{(p)}_I\omega}=\frac{1}{d}\Trace{(\rho^{(p)}_I)^{T_2}\mathbb{F}_d}\in\lbr 0, 1\rbr$ and normalized to $\Trace{\rho^{(p)}_I}=1$. These states are entangled (and NPPT) for $p\in(\frac{1}{d},1]$, and separable for $p\in[0, \frac{1}{d}]$ (thus, PPT). It is well-known that all entangled isotropic states are distillable~\cite{PhysRevA.59.4206}.

To obtain distillation schemes as needed in our proofs we will apply suitable twirling operations to the Choi-matrix $C_\Pm$ of a positive map $\Pm:\M_d\ra\M_d$. The matrix $C_\Pm$ is in general not positive, but the next lemma proves that under certain conditions the $UU$-twirl leads to positive matrix, which can then be distilled using the existing theory.

\begin{lem}[Twirl of block-positive matrices]

If $C\in\M_{d}\otimes\M_d$ is block-positive and such that $\Trace{C\mathbb{F}_d}\leq 0$, then 
\begin{align}
\int_{U\in\Um(d)} \lb U\otimes U\rb C \lb U\otimes U\rb^\dagger \text{dU} \geq 0.
\label{equ:posAfterTwirl1}
\end{align}Similarly, if $C\in\M_{d^2}$ is block-positive and such that $\Trace{C\omega_d}\geq 0$, then 
\begin{align}
\int_{U\in\Um(d)} \lb U\otimes \overline{U}\rb C \lb U\otimes \overline{U}\rb^\dagger \text{dU} \geq 0.
\label{equ:posAfterTwirl2}
\end{align}

\label{lem:TwirlBlock}
\end{lem}

\begin{proof}

By block-positivity we have $\Trace{C}\geq 0$ and
\begin{align*}
\Trace{C\lb \id_d\otimes\id_d + \mathbb{F}_d\rb} \geq 0
\end{align*} 
because the Werner state $\rho^{(1)}_W=(\id_d\otimes\id_d + \mathbb{F}_d)/(d(d+1))$ is separable. Together with the assumption $\Trace{C\mathbb{F}_d}\leq 0$ this implies $\Trace{C\mathbb{F}_d}\in\lbr -\Trace{C},0\rbr$ which shows the first statement \eqref{equ:posAfterTwirl1}.

Secondly, as the Werner state $\rho^{(0)}_W$ is also separable we get by block-positivity
\begin{align*}
\Trace{C^{T_2}\lb \id_d\otimes\id_d - \frac{1}{d}\mathbb{F}_d\rb} \geq 0,
\end{align*}
which implies $d\Trace{C^{T_2}}\geq \Trace{C^{T_2}\mathbb{F}_d}$. Together with the assumption $\Trace{C^{T_2}\mathbb{F}_d} = d\Trace{C\omega_d}\geq 0$ this implies that the $UU$-twirl of $C^{T_2}$ has positive partial transpose. As $\lbr\lb U\otimes U\rb X^{T_2}\lb U\otimes U\rb^\dagger\rbr^{T_2} = \lb U\otimes \overline{U}\rb X\lb U\otimes \overline{U}\rb^\dagger$ this finishes the proof.

\end{proof}

\section{Minimal output eigenvalue}
\label{Appendix2}

Here we prove the multiplicativity of the minimal output eigenvalue \eqref{minoutev} for entanglement breaking completely positive maps in an elementary way. Afterwards we outline the proof of a more general statement, extending a multiplicativity result by King \cite{king2002maximal}.

\begin{thm}

Let $\Tm:\M_{d_1}\ra\M_{d_2}$ be entanglement breaking, i.e. $(\ident_n\otimes \Tm)\lb\rho\rb$ is separable for all $n\in\N$ and positive $\rho\in\M_{n}\otimes \M_{d_1}$, and $\Sm:\M_{d_3}\ra \M_{d_4}$ be completely positive. Then we have 
\begin{align*}
\lambda^{\text{min}}_{\text{out}}\lbr \Tm\otimes \Sm\rbr = \lambda^{\text{min}}_{\text{out}}\lbr \Tm\rbr \lambda^{\text{min}}_{\text{out}}\lbr \Sm\rbr.
\end{align*}
\label{Additivity}
\end{thm}
\begin{proof}

By inserting product states it is clear that $\lambda^{\text{min}}_{\text{out}}\lbr \Tm\otimes \Sm\rbr \leq  \lambda^{\text{min}}_{\text{out}}\lbr \Tm\rbr \lambda^{\text{min}}_{\text{out}}\lbr \Sm\rbr$.

For the other direction, let the minimum in \eqref{minoutev} for the computation of $\lambda^{\text{min}}_{\text{out}}\lbr \Tm\otimes \Sm\rbr$ be attained at $\rho=\tau$. Then there exists a pure state $\ket{\phi}$ such that
\begin{align*}
\lambda^{\text{min}}_{\text{out}}\lbr \Tm\otimes \Sm\rbr &= \bra{\phi} (\Tm\otimes \Sm)(\tau) \ket{\phi}\\
&= \bra{\phi}\sum^k_{i=1} \lbr\sigma_{i}\otimes \Sm(\rho_i)\rbr\ket{\phi}
\end{align*}
using that there exist non-zero $\sigma_i\geq 0$ and $\rho_i \geq 0$ such that $(\Tm\otimes\ident_{d_3})\lb \tau\rb = \sum^k_{i=1} \sigma_i \otimes \rho_i$ as $\Tm$ is entanglement breaking. Note that $\Tm\lb\text{tr}_2\lb\tau\rb\rb = \sum^k_{i=1} \text{tr}\lb\rho_i\rb\sigma_i$. Thus:
\begin{align*}
\lambda^{\text{min}}_{\text{out}}\lbr \Tm\otimes \Sm\rbr &= \bra{\phi}\sum^k_{i=1} \lbr\text{tr}(\rho_i) \sigma_{i}\otimes \Sm(\frac{\rho_i}{\text{tr}\lb\rho_i\rb})\rbr\ket{\phi} \\
&\geq \lambda^{\text{min}}_{\text{out}}\lbr \Sm\rbr\bra{\phi}\sum^k_{i=1} \text{tr}\lb\rho_i\rb\sigma_{i}\otimes  \id_{d_4}\ket{\phi} \\
&= \lambda^{\text{min}}_{\text{out}}\lbr \Sm\rbr\text{tr}\lbr \Tm\lb\text{tr}_2\lb\tau\rb\rb\text{tr}_2\lb\proj{\phi}{\phi}\rb\rbr \\
&\geq\lambda^{\text{min}}_{\text{out}}\lbr \Sm\rbr\lambda_\text{min}\big(\Tm\lb\text{tr}_2\lb\tau\rb\rb\big)\\
&\geq \lambda^{\text{min}}_{\text{out}}\lbr \Sm\rbr\lambda^{\text{min}}_{\text{out}}\lbr \Tm\rbr.  
\end{align*}
\end{proof}

It may be of interest that Theorem \ref{Additivity} can also be obtained from a generalization of King's result \cite{king2002maximal} on the multiplicativity of output $p$-norms for entanglement-breaking completely positive maps (note that King's proof does not require the trace-preservation property, so the result does not only hold for quantum channels). To set up notation, define for any completely positive map $\Sm$ and for any $p\in(-\infty,+\infty)$ the following quantities:
\begin{align*}
\mu_p^+(\Sm)\,&:=\,\sup_\rho\tr[\Sm(\rho)^p]\,,\\
\mu_p^-(\Sm)\,&:=\,\inf_\rho\tr[\Sm(\rho)^p]\,,
\end{align*}
where the optimizations run over all quantum states $\rho$ at the input of the map $\Sm$. By the same derivation as in \cite{king2002maximal} and noting that the Lieb-Thirring inequality holds for all $p\in(-\infty,-1]\cup[1,\infty)$, whereas the inequality sign can be reversed for $p\in[-1,+1]$, one can prove the following:
\begin{thm}For any completely positive map $\Sm$ and any entanglement-breaking completely positive map $\Tm$:
\begin{align*}
\mu_p^+(\Tm\otimes\Sm)\,&=\,\mu_p^+(\Tm)\,\mu_p^+(\Sm)\qquad\text{for}~p\in(-\infty,-1]\cup[1,\infty)\,,\\
\mu_p^-(\Tm\otimes\Sm)\,&=\,\mu_p^-(\Tm)\,\mu_p^-(\Sm)\qquad\text{for}~p\in[-1,+1]\,.
\end{align*}
\end{thm}To see that this implies Theorem \ref{Additivity}, note first that for $p\in(-\infty,\infty)\setminus\{0\}$ a similar multiplicativity result then holds for the quantities $\big(\mu^{\pm}_p(\cdot)\big)^{1/p}$ as well. Finally, in the limit $p\to-\infty$ we have, for any completely positive map $\Sm$,
\begin{align*}
\lim_{p\to-\infty}\big(\mu^+_p(\Sm)\big)^{1/p}\,&=\,\lim_{p\to-\infty}\big(\sup_\rho\tr[\Sm(\rho)^p]\big)^{1/p}\\
&=\,\lim_{p\to-\infty}\inf_\rho\big(\tr[\Sm(\rho)^p]\big)^{1/p}\,=\,\inf_\rho\lambda_{\min}(\Sm(\rho))\,=\,\lambda^{\text{min}}_{\text{out}}\lbr \Sm\rbr\,.
\end{align*}

One can translate these multiplicativity results to the language of minimum respectively maximum output Renyi $p$-entropies, which are defined for $p\in(-\infty,\infty)\setminus\{1\}$ by $H^{\min}_p(\Sm):=\inf_\rho\frac{1}{1-p}\log\tr[\Sm(\rho)^p]$ and $H^{\max}_p(\Sm):=\sup_\rho\frac{1}{1-p}\log\tr[\Sm(\rho)^p]$, and extended by continuity to the values $p=1,\pm\infty$ as $H^{\min}_{p=1}(\Sm):=\inf_\rho\tr[-\Sm(\rho)\log\Sm(\rho)]$ (the minimum output von Neumann entropy), $H^{\min}_{p=\infty}(\Sm):=\inf_\rho\big(-\log\lambda_{\max}(\Sm(\rho))\big)$ (the minimum output min entropy), and $H^{\max}_{p=-\infty}(\Sm):=\sup_\rho\big(-\log\lambda_{\min}(\Sm(\rho))\big)=-\log\lambda^{\text{min}}_{\text{out}}\lbr \Sm\rbr$. For a proper interpretation as output entropies, the map $\Sm$ should, in addition to being (completely) positive, also be trace-preserving, although this is not a mathematical requirement. We thus obtain from the above:
\begin{cor}For any completely positive map $\Sm$ and any entanglement-breaking completely positive map $\Tm$:
\begin{align*}
H^{\min}_p(\Tm\otimes\Sm)\,&=\,H^{\min}_p(\Tm)+H^{\min}_p(\Sm)\qquad\text{for}~p\in[-1,\infty]\,,\\
H^{\max}_p(\Tm\otimes\Sm)\,&=\,H^{\max}_p(\Tm)+H^{\max}_p(\Sm)\qquad\text{for}~p\in[-\infty,-1]\,.
\end{align*}
\end{cor}
For trace-preserving $\Sm,\Tm$ and $p=1$, this result was first obtained by Shor~\cite{shor2002additivity}, whereas the case $p\in(1,\infty)$ follows from \cite{king2002maximal}.

\bibliographystyle{abbrv}

\begin{thebibliography}{10}

\bibitem{PhysRevLett.82.5385}
C.~H. Bennett, D.~P. DiVincenzo, T.~Mor, P.~W. Shor, J.~A. Smolin, and B.~M.
  Terhal.
\newblock Unextendible product bases and bound entanglement.
\newblock {\em Phys. Rev. Lett.}, 82:5385--5388, Jun 1999.

\bibitem{bennett1997capacities}
C.~H. Bennett, D.~P. DiVincenzo, and J.~A. Smolin.
\newblock Capacities of quantum erasure channels.
\newblock {\em Physical Review Letters}, 78(16):3217, 1997.

\bibitem{berta2012entanglement}
M.~Berta, M.~Christandl, F.~G. Brandao, and S.~Wehner.
\newblock Entanglement cost of quantum channels.
\newblock In {\em Information Theory Proceedings (ISIT), 2012 IEEE
  International Symposium on}, pages 900--904. IEEE, 2012.

\bibitem{bhatia1997matrix}
R.~Bhatia.
\newblock {\em Matrix Analysis}.
\newblock Graduate Texts in Mathematics. Springer New York, 1997.

\bibitem{bruss1998optimal}
D.~Bru{\ss}, D.~P. DiVincenzo, A.~Ekert, C.~A. Fuchs, C.~Macchiavello, and
  J.~A. Smolin.
\newblock Optimal universal and state-dependent quantum cloning.
\newblock {\em Physical Review A}, 57(4):2368, 1998.

\bibitem{chitambar2014everything}
E.~Chitambar, D.~Leung, L.~Man{\v{c}}inska, M.~Ozols, and A.~Winter.
\newblock Everything you always wanted to know about locc (but were afraid to
  ask).
\newblock {\em Communications in Mathematical Physics}, 328(1):303--326, 2014.

\bibitem{6094278}
T.~Cubitt, J.~Chen, and A.~Harrow.
\newblock Superactivation of the asymptotic zero-error classical capacity of a
  quantum channel.
\newblock {\em Information Theory, IEEE Transactions on}, 57(12):8114--8126,
  Dec 2011.

\bibitem{PhysRevA.61.062312}
D.~P. DiVincenzo, P.~W. Shor, J.~A. Smolin, B.~M. Terhal, and A.~V. Thapliyal.
\newblock Evidence for bound entangled states with negative partial transpose.
\newblock {\em Phys. Rev. A}, 61:062312, May 2000.

\bibitem{PhysRevA.61.062313}
W.~D\"ur, J.~I. Cirac, M.~Lewenstein, and D.~Bru\ss{}.
\newblock Distillability and partial transposition in bipartite systems.
\newblock {\em Phys. Rev. A}, 61:062313, May 2000.

\bibitem{filippov2013dissociation}
S.~N. Filippov, A.~A. Melnikov, and M.~Ziman.
\newblock Dissociation and annihilation of multipartite entanglement structure
  in dissipative quantum dynamics.
\newblock {\em Physical Review A}, 88(6):062328, 2013.

\bibitem{filippov2012local}
S.~N. Filippov, T.~Ryb{\'a}r, and M.~Ziman.
\newblock Local two-qubit entanglement-annihilating channels.
\newblock {\em Physical Review A}, 85(1):012303, 2012.

\bibitem{filippov2013bipartite}
S.~N. Filippov and M.~Ziman.
\newblock Bipartite entanglement-annihilating maps: Necessary and sufficient
  conditions.
\newblock {\em Physical Review A}, 88(3):032316, 2013.

\bibitem{hayashiQIbook}
M.~Hayashi.
\newblock {\em Quantum Information -- An Introduction}.
\newblock Springer Berlin Heidelberg, 2006.

\bibitem{holevo2001evaluating}
A.~S. Holevo and R.~F. Werner.
\newblock Evaluating capacities of bosonic gaussian channels.
\newblock {\em Physical Review A}, 63(3):032312, 2001.

\bibitem{PhysRevA.59.4206}
M.~Horodecki and P.~Horodecki.
\newblock Reduction criterion of separability and limits for a class of
  distillation protocols.
\newblock {\em Phys. Rev. A}, 59:4206--4216, Jun 1999.

\bibitem{horodecki1996separability}
M.~Horodecki, P.~Horodecki, and R.~Horodecki.
\newblock Separability of mixed states: necessary and sufficient conditions.
\newblock {\em Physics Letters A}, 223(1):1--8, 1996.

\bibitem{PhysRevLett.80.5239}
M.~Horodecki, P.~Horodecki, and R.~Horodecki.
\newblock Mixed-state entanglement and distillation: Is there a ``bound''
  entanglement in nature?
\newblock {\em Phys. Rev. Lett.}, 80:5239--5242, Jun 1998.

\bibitem{horodecki2003entanglement}
M.~Horodecki, P.~W. Shor, and M.~B. Ruskai.
\newblock Entanglement breaking channels.
\newblock {\em Reviews in Mathematical Physics}, 15(06):629--641, 2003.

\bibitem{king2002maximal}
C.~King.
\newblock Maximal p-norms of entanglement breaking channels.
\newblock {\em Quantum Information {\&} Computation}, 3(2):186--190, 2003.

\bibitem{kretschmann2004tema}
D.~Kretschmann and R.~F. Werner.
\newblock Tema con variazioni: quantum channel capacity.
\newblock {\em New Journal of Physics}, 6(1):26, 2004.

\bibitem{PhysRevA.55.1613}
S.~Lloyd.
\newblock Capacity of the noisy quantum channel.
\newblock {\em Phys. Rev. A}, 55:1613--1622, Mar 1997.

\bibitem{moravvcikova2010entanglement}
L.~Morav{\v{c}}{\'{\i}}kov{\'a} and M.~Ziman.
\newblock Entanglement-annihilating and entanglement-breaking channels.
\newblock {\em Journal of Physics A: Mathematical and Theoretical},
  43(27):275306, 2010.

\bibitem{morganwinter}
C.~Morgan and A.~Winter.
\newblock ``{P}retty strong'' converse for the quantum capacity of degradable
  channels.
\newblock {\em Information Theory, IEEE Transactions on}, 60(1):317--333, Jan
  2014.

\bibitem{paulsen2002completely}
V.~Paulsen.
\newblock {\em Completely bounded maps and operator algebras}, volume~78.
\newblock Cambridge University Press, 2002.

\bibitem{rains2001semidefinite}
E.~M. Rains.
\newblock A semidefinite program for distillable entanglement.
\newblock {\em Information Theory, IEEE Transactions on}, 47(7):2921--2933,
  2001.

\bibitem{roy2009unitary}
A.~Roy and A.~J. Scott.
\newblock Unitary designs and codes.
\newblock {\em Designs, codes and cryptography}, 53(1):13--31, 2009.

\bibitem{shor2002additivity}
P.~W. Shor.
\newblock Additivity of the classical capacity of entanglement-breaking quantum
  channels.
\newblock {\em Journal of Mathematical Physics}, 43(9):4334--4340, 2002.

\bibitem{PhysRevLett.108.230507}
G.~Smith and J.~A. Smolin.
\newblock Detecting incapacity of a quantum channel.
\newblock {\em Phys. Rev. Lett.}, 108:230507, Jun 2012.

\bibitem{stormer2010tensor}
E.~St{\o}rmer.
\newblock Tensor powers of 2-positive maps.
\newblock {\em J. of Math. Phys}, 51(10), 2010.

\bibitem{6832533}
M.~Takeoka, S.~Guha, and M.~Wilde.
\newblock The squashed entanglement of a quantum channel.
\newblock {\em Information Theory, IEEE Transactions on}, 60(8):4987--4998, Aug
  2014.

\bibitem{tomamichelwildewinter}
M.~Tomamichel, M.~M. Wilde, and A.~Winter.
\newblock Strong converse rates for quantum communication.
\newblock {\em arXiv:1406.2946}, 2014.

\bibitem{werner1989quantum}
R.~F. Werner.
\newblock Quantum states with einstein-podolsky-rosen correlations admitting a
  hidden-variable model.
\newblock {\em Physical Review A}, 40(8):4277, 1989.

\end{thebibliography}

\end{document}